\documentclass[journal]{IEEEtran}
\usepackage{amsmath, amssymb, cite, tikz, graphicx, xspace, mathrsfs, mathtools, psfrag, chemarrow, subfigure,color,soul,graphicx,mathtools,adjustbox}
\usepackage{kbordermatrix}
\usetikzlibrary{automata, positioning}
\usepackage{standalone}
\usepackage{enumitem}
\usepackage{bbm}
\usepackage[utf8]{inputenc}
\usepackage{booktabs}
\usepackage{amscd}
\usepackage{amsmath}
\usepackage{amssymb}
\usepackage{amsthm}

\usepackage{epsfig}
\usepackage{verbatim}
\usepackage{graphicx}
\usepackage{amsthm}
\usepackage{color}
\usepackage{ multirow }
\usepackage[noend]{algpseudocode}
\usetikzlibrary{patterns}
\usetikzlibrary{angles} 
\usepackage{eucal}
\usepackage{algorithm}
\usepackage{tikz}
\usetikzlibrary{shapes, arrows, decorations.markings, arrows.meta}
\usetikzlibrary{patterns} 
\usetikzlibrary{arrows,automata,positioning,chains,calc}
\usetikzlibrary{quotes,angles}
\newtheorem{remark}{Remark}
\makeatletter
\def\BState{\State\hskip-\ALG@thistlm}
\makeatother
\newtheorem{theorem}{Theorem}
\newtheorem{lemma}[theorem]{Lemma}
\begin{document}
	\setlength{\abovecaptionskip}{-3pt}
	\setlength{\belowcaptionskip}{1pt}
	\setlength{\floatsep}{1ex}
	\setlength{\textfloatsep}{1ex}	
	\title{Optimal Sampling and Actuation for Real-Time Monitoring of Markov Sources}
    \author{Mehrdad Salimnejad, Anthony Ephremides,  \IEEEmembership{Life Fellow, IEEE}, Marios Kountouris, \IEEEmembership{Fellow, IEEE}, and Nikolaos Pappas, 
\IEEEmembership{Senior Member, IEEE}.
\thanks{\scriptsize M. Salimnejad and N. Pappas are with the Department of Computer and Information Science, Linköping University, Sweden, email: \{\texttt{mehrdad.salimnejad, nikolaos.pappas\}@liu.se}. A. Ephremides is with the Department of Electrical and Computer Engineering, University of Maryland, College Park, MD, USA, email: \texttt{etony@umd.edu}. M. Kountouris is with the Department of Computer Science and Artificial Intelligence, Andalusian Research Institute in Data Science and Computational Intelligence (DaSCI), University of Granada, Spain, email: \texttt{mariosk@ugr.es}. This work has been supported in part by ELLIIT, the European Union (ROBUST-6G, 101139068,  6G-LEADER, 101192080, and SOVEREIGN, 101131481), and the European Research Council (ERC) SONATA Project (Grant agreement No. 101003431). A shorter version of this work has been published in \cite{salimnejad2026optimalsamplingactuationpolicies}.}}

	\maketitle

\begin{abstract}
This paper studies efficient data management and timely information dissemination for real-time monitoring of an $N$-state Markov process, with the objective of enabling accurate state estimation and reliable actuation decisions. 
We analyze the real-time reconstruction error and the Age of Incorrect Information (AoII), and derive closed-form expressions for their time-averaged values under several sampling and transmission policies. We then formulate and solve constrained optimization problems to minimize the time-averaged reconstruction error and the average AoII under a time-averaged sampling frequency constraint. The resulting optimal sampling and transmission policies are compared to identify the conditions under which each policy is most effective. We further show that directly using the reconstructed state for actuation can degrade system performance, especially when the receiver is uncertain about the state estimate or when actuation is costly. These findings reveal that accurate state estimation alone does not necessarily lead to effective actuation, highlighting the importance of incorporating uncertainty into the decision-making process. To address this issue, we introduce a cost function, termed the \emph{Cost of Actions under Uncertainty} (CoAU), which characterizes correct and incorrect actuation decisions under receiver-side uncertainty. We propose a \emph{randomized actuation policy} and derive a closed-form expression for the probability of a correct actuation decision, defined as the event in which the CoAU equals zero. Finally, we formulate an optimization problem to find the optimal randomized actuation policy that maximizes this probability. The results show that the resulting policy substantially reduces incorrect actuator actions.
	\end{abstract}
	
\section{Introduction}
Real-time networked monitoring systems increasingly support decisions that have direct physical consequences, such as triggering medical alerts, coordinating autonomous vehicles, regulating industrial processes, or controlling cyber-physical infrastructures~\cite{hult2016coordination,vitturi2013industrial,pappas2022agebook,shreedhar2019age}.  The effectiveness of these applications depends not only on receiving status updates but also on ensuring that the received information is timely, relevant, and actionable for decision-making and actuation.

As these systems operate under limited energy, bandwidth, and computational resources, transmitting every observed update is often inefficient and may even degrade system performance. Many generated updates are either redundant, weakly informative, or irrelevant to the downstream task. This has motivated a shift from conventional data-centric communication toward goal-oriented and semantics-aware communication, where updates are generated and transmitted according to their relevance to the system objective~\cite{kountouris2021semantics,kalfa2021towards,popovski2020semantic,popovski2022perspective,gunduz2022beyond,strinati2024goal,luo2025informationfreshnesssemanticsinformation}. By avoiding unnecessary transmissions, such policies improve resource efficiency and scalability while preserving the information needed for timely decision-making.

The Age of Information (AoI) is a widely used metric for quantifying the timeliness, or freshness, of received information~\cite{kaul2012real}. It is defined as the time elapsed since the generation of the most recently received update. However, freshness alone does not determine whether the receiver has an accurate or useful representation of the source. To overcome this limitation, the Age of Incorrect Information (AoII) was introduced in~\cite{maatouk2020age} to measure the duration for which the system operates with an incorrect state estimate, thereby combining timeliness and estimation accuracy. Beyond AoI and AoII, several semantics-aware metrics have been proposed to capture the relevance, significance, and task effectiveness of transmitted information~\cite{pappas2021goal,salimnejad2024real,MSalimnejadJCN2023,luo2024semantic,yates2021age,cocco2023remote,MSalimnejadTCOM2025,delfani2024semantics,AliNikkhahTCOM2026}. 

Although these metrics substantially improve the design of sampling and transmission policies, they primarily assess the quality of information at the receiver. In actuation-oriented systems, however, accurate reconstruction of the source state may not be the final objective. The actuator must also decide whether the reconstructed state should be trusted and used to trigger an action. This distinction is critical because directly acting on an uncertain or outdated reconstructed state can lead to incorrect actuation decisions, especially when the channel is unreliable, the sampling rate is constrained, or actuation is costly. Conversely, remaining idle may be preferable when the receiver cannot reliably infer the correctness of the reconstructed source state. Therefore, minimizing estimation error or AoII alone does not necessarily guarantee reliable actuation. This motivates the design of actuation-aware policies that explicitly account for receiver-side uncertainty.

To address this issue, we introduce a cost function, termed the \emph{Cost of Actions under Uncertainty} (CoAU), which characterizes incorrect actuation decisions under receiver-side uncertainty. In particular, CoAU captures two types of undesirable behavior: acting when the reconstructed state is incorrect and remaining idle when the reconstructed state is correct. Accordingly, the event that CoAU equals zero corresponds to a correct actuation decision, which includes either taking an action when the estimate is correct or avoiding an action when the estimate is incorrect. This formulation enables the actuator to balance the benefit of acting against the risk of acting under uncertainty.

In this paper, we study a real-time remote monitoring system in which an $N$-state Markov source is observed by a sampler and transmitted over an unreliable wireless channel. At each time slot, the sampler decides whether to generate a sample. If a sample is generated, the transmitter sends it to the receiver; otherwise, it remains idle. The receiver reconstructs the source state from successfully received samples, and the actuator decides whether to perform an action based on the reconstructed state and the available uncertainty information. We consider randomized stationary, change-aware randomized stationary, semantics-aware randomized stationary, and threshold-aware randomized stationary sampling and transmission policies. For these policies, we derive closed-form expressions for the time-averaged reconstruction error and the average AoII. We then formulate constrained optimization problems to minimize these metrics subject to a time-averaged sampling-frequency constraint. Finally, we propose a \emph{randomized actuation} policy and optimize its parameters to maximize the probability that the CoAU equals zero under a time-averaged actuation-frequency constraint. These results demonstrate that effective real-time monitoring requires the joint design of sampling, transmission, and actuation policies.

\subsection{Related Works}
Our work considers the problem of designing scheduling policies for data generation, transmission, and actuation to enhance the performance of real-time remote monitoring of Markovian processes. Several studies have examined optimal communication and estimation policies for Markovian systems \cite{nayyar2013optimal,chakravorty2014optimal,shi2012scheduling,wu2018optimal,chakravorty2019remote}. In different scenarios—such as energy-limited sources \cite{nayyar2013optimal}, costly transmissions \cite{chakravorty2014optimal}, bandwidth-constrained multi-sensor systems \cite{shi2012scheduling,wu2018optimal}, and power-controlled transmission over time-varying channels with feedback \cite{chakravorty2019remote}—these studies show that optimal strategies usually have simple and structured forms, including threshold-based, periodic, or monotone policies. However, these works generally treat all information as equally important and focus on minimizing estimation error, without explicitly considering the timeliness, relevance, or task-specific importance of the information, or its effect on actuation performance. The works 
\cite{pappas2021goal,salimnejad2024real,MSalimnejadJCN2023,jayanth23,fountoulakis2023goal,luo2024semantic,luo2025cost,cocco2023remote,santi2024remote,talli2024pragmatic,talli2024push,yates2021age,buyukatesversion,kaswan2022timely,mitra2023age,KaswanTCOM2023,KaswanJSAC2023,delfani2023version,MSalimnejadTCOM2025,MSalimnejadWCNC2025,holm2021freshness,delfani2024semantics,poojary2017real,jiang2019status,zhou2020age,kalor2019minimizing,ramakanth2024monitoring,salimnejad2025real}
investigate the goal-oriented effectiveness of communicated updates by optimizing semantics-aware performance metrics across the entire information chain, from data generation to transmission and actuation. The work \cite{pappas2021goal} introduces the cost of actuation error to quantify the impact of incorrect actions at the receiver, explicitly linking communication decisions to control performance. Building on this perspective, joint sampling and transmission policies for Markovian sources were developed in \cite{salimnejad2024real, MSalimnejadJCN2023}, where consecutive error and cost of memory error \cite{salimnejad2024real} capture error persistence, and importance-aware consecutive error \cite{MSalimnejadJCN2023} incorporates task relevance. The urgency of persistent mismatches was further modeled through the Age of Consecutive Error (AoCE) in \cite{luo2025cost}, which leads to structured threshold-based optimal policies. Optimal strategies under communication and resource constraints have also been widely studied. Distortion minimization under AoI and source costs was analyzed in \cite{jayanth23}, actuation-cost-aware policies were developed in \cite{fountoulakis2023goal} and extended to multi-source systems in \cite{luo2024semantic}, and constrained-channel monitoring was examined in \cite{cocco2023remote, santi2024remote}. Push- and pull-based architectures were investigated in \cite{talli2024pragmatic, talli2024push}, highlighting tradeoffs between control performance and communication overhead.  In addition, extensions of the AoI framework incorporate source evolution and the relevance of information. The Version Age of Information (VAoI) was introduced in \cite{yates2021age} to measure the version lag between source and monitor, and its minimization in gossip networks was studied in \cite{buyukatesversion, kaswan2022timely, mitra2023age, KaswanTCOM2023, KaswanJSAC2023, delfani2023version}. For Markovian sources, \cite{MSalimnejadTCOM2025} proposed Version Innovation Age (VIA) and Age on Incorrect Version (AoIV), enabling semantics-aware sampling and transmission policies, with energy harvesting extensions analyzed in \cite{MSalimnejadWCNC2025}. In pull-based settings, the Query Age of Information  (QAoI) was introduced in \cite{holm2021freshness}, and the Query Version AoI (QVAoI) was proposed in \cite{delfani2024semantics} to jointly capture freshness and importance at query instances. Furthermore, the monitoring of correlated sources under temporal and spatial dependencies has been studied in \cite{poojary2017real, jiang2019status, zhou2020age, kalor2019minimizing, ramakanth2024monitoring, salimnejad2025real}, addressing challenges such as correlation-aware scheduling and AoI or reconstruction-error minimization. However, existing studies—even in semantics-aware frameworks—largely overlook how receiver-side uncertainty and actuation costs affect downstream decision-making, often assuming that actions can be directly and safely taken based on reconstructed states, which can result in significant performance degradation in actuation-oriented systems.

Despite these advances, most existing semantics-aware frameworks optimize the quality of reconstructed information rather than the reliability of the subsequent actuation decision. In particular, they often assume that the actuator can directly use the reconstructed state or the most recently received update. This assumption may be inappropriate when the receiver is uncertain about whether the reconstructed state matches the true source state. The present work addresses this gap by jointly considering sampling, transmission, receiver-side uncertainty, and randomized actuation.

\subsection{Contributions}
The main contributions of this paper are summarized as follows:
\begin{enumerate}
    \item We derive closed-form expressions for the time-averaged reconstruction error and the average AoII for an $N$-state Markov source under randomized stationary, change-aware randomized stationary, and semantics-aware randomized stationary joint sampling and transmission policies proposed in \cite{salimnejad2024real,MSalimnejadTCOM2025}. For the AoII metric, we further introduce a \emph{threshold-aware randomized stationary} policy that prioritizes sampling when incorrect information persists at the receiver. 
    \item We formulate constrained optimization problems for minimizing the time-averaged reconstruction error and the average AoII under a time-averaged sampling-frequency constraint. The resulting solutions characterize the operating regimes in which each sampling and transmission policy is most effective.
    \item We introduce the \emph{Cost of Actions under Uncertainty} (CoAU), an actuation-aware cost function that captures incorrect actuation decisions under receiver-side uncertainty. Based on this metric, we propose a randomized actuation policy in which the actuator probabilistically decides whether to act depending on whether the receiver is certain or uncertain about the reconstructed state.
    \item We derive closed-form expressions for the probability that the CoAU equals zero under the proposed randomized actuation policy. We then formulate and solve a constrained optimization problem that maximizes this probability subject to a time-averaged actuation-frequency constraint. The results demonstrate that the optimized randomized actuation policy can substantially reduce incorrect actuation decisions compared with an always-act policy.
\end{enumerate}

\section{System Model}
\label{SystemModel}
We consider a remote monitoring system that includes an information source, a sampler, a transmitter, a receiver, and an actuator, as shown in Fig. 1. The system operates over discrete time, which we represent as a sequence of equally spaced time slots indexed by $t\in \mathbb{N}$. The state of the information source at time slot $t$ is denoted by $X(t)$, and it is modeled as an $N$-state discrete-time Markov chain (DTMC). This state represents information that can be mapped to an action at the receiver side, or it can be considered as a control set that determines the actuation at the actuation point. Specifically, the self-transition probability and the probability of transitioning to a different state at time slot $t+1$ are defined as $\mathbb{P}\big[X(t+1)=i \big|X(t)=i\big]=q$ and $\mathbb{P}\big[X(t+1)=j \big|X(t)=i\big]=p, \forall i,j \in\{0,1,\cdots, N-1\}$ with $i\neq j$, respectively, where $q+(N-1)p=1$. 
\par At the beginning of each time slot $t$, the sampler observes the current state of the information source and decides whether to perform sampling. The sampling action at time slot $t$ is represented by $\alpha(t) \in \{0,1\}$, where $\alpha(t) = 1$ indicates that the source is sampled and $\alpha(t) = 0$ indicates that no sampling occurs. When the source is sampled, the transmitter immediately sends the sample as a packet over the wireless communication channel to the receiver; otherwise, the transmitter remains idle. We assume that the transmission is successful when the channel state $h(t)$ equals $1$, while $h(t)=0$ indicates a decoding failure. The probability of successful decoding is defined as $p_s = \mathbb{P}[h(t) = 1]$. 
At the end of time slot $t$, the receiver constructs an estimate of the information source state, denoted by $\hat{X}(t)$, based on the successfully received samples. If the receiver successfully decodes the transmitted sample, then $X(t)=\hat{X}(t)$; otherwise, the reconstructed state remains unchanged, i.e., $\hat{X}(t) = \hat{X}(t-1)$. Likewise, if the transmitter does not send a sample or if the transmission fails, it is assumed that $\hat{X}(t) = \hat{X}(t-1)$. At time slot $t$, the system is considered to be in a synchronized (synced) state if $X(t) = \hat{X}(t)$, and in an erroneous state if $X(t) \neq \hat{X}(t)$. Acknowledgment (ACK) and negative-acknowledgment (NACK) packets are used to inform the transmitter about the success or failure of transmissions, and these feedback packets are assumed to be delivered instantly and without errors. Therefore, the transmitter has perfect knowledge of the receiver's reconstructed state at each time slot. In addition, we assume that any sample from an unsuccessful transmission is discarded. The actuator then decides whether to perform an action based on the reconstructed source state. The actuator’s action at time slot $t$ is denoted by $c(t) \in \{0,1\}$, where $c(t)=1$ indicates that the actuator performs an action, and $c(t)=0$ indicates that it remains idle.  Here, for sampling and transmission, we consider three existing policies: \emph{randomized stationary}, \emph{change-aware randomized stationary}, and \emph{semantics-aware randomized stationary}, as introduced in \cite{salimnejad2024real} and \cite{MSalimnejadTCOM2025}. A brief description of each policy is given below.
\begin{enumerate}
   \item  Randomized Stationary (RS): In this policy, a new sample is generated at each time slot $t$ with probability $p^{r}_{\alpha}$, regardless of whether the system is synced. 
    \item Change-aware randomized stationary (CARS): In this policy, at time slot $t$, no sample is taken if the source state has not changed, i.e., $X(t)=X(t-1)$. Otherwise, if $X(t) \neq X(t-1)$, a new sample is generated with probability $p^{c}_{\alpha}$.
    \item Semantics-aware randomized stationary (SARS): In this policy, at each time slot, the generation of a new sample is triggered probabilistically only when the system is in an incorrect state, i.e., $X(t)\!\neq\! \hat{X}(t-1)$. Specifically, if $X(t-1)\!\!=\!\!\hat{X}(t-1)$ and $X(t)\!\!\neq\!\! X(t-1)$, the sampler generates a sample at time slot $t$ with probability $q_{\alpha_{1}}$. Furthermore, if $X(t-1)\!\!\neq\!\! \hat{X}(t-1)$ and  $X(t)\!\!\neq \!\!\hat{X}(t\!-\!1)$, the sampler generates a sample at time slot $t$ with probability $q_{\alpha_{2}}$, where $q_{\alpha_{2}}\geqslant q_{\alpha_{1}}$.
\end{enumerate}

\begin{figure}[t!]
    \centering
    \includegraphics[width=1\linewidth]{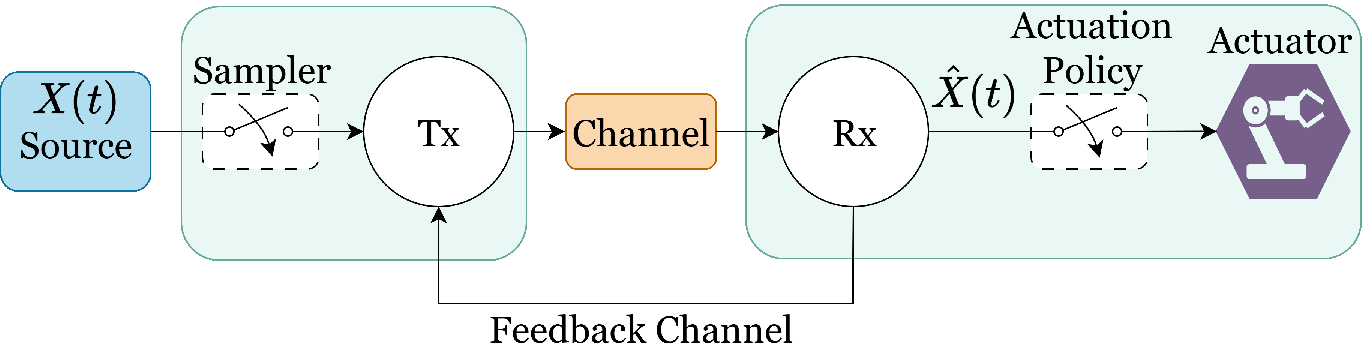}
    \caption{Real-time monitoring of a Markovian source over a wireless channel.}
    \label{system_model_fig}
\end{figure}  
\section{Performance Metrics and Analysis}
\par In this section, we investigate the impact of incorrect information and actions at the receiver and actuator, respectively. We first evaluate two metrics—\emph{time-averaged reconstruction error} and \emph{Age of Incorrect Information} (AoII)—to measure reconstruction accuracy at the receiver. However, making actuation decisions directly based on the reconstructed states may pose safety risks, particularly when the receiver is uncertain about the correctness of the reconstructed source state or when actuation incurs high costs. To address this, we introduce the \emph{Cost of Actions under Uncertainty} (CoAU), a metric that explicitly quantifies incorrect actuation decisions.
\subsection{Real-Time Reconstruction Error}The real-time reconstruction error at time slot $t$ is defined as the difference between $X(t)$ and $\hat{X}(t)$, i.e.,
\begin{align}
    \label{Error}
    E(t) = \big|X(t)-\hat{X}(t)\big|,
\end{align}
The time-averaged reconstruction error over the observation interval $[1, T]$, where $T$ is a large positive number, is calculated as:
\begin{align}
    \label{PE}
    \bar{E} =  \lim_{T \to \infty} \frac{1}{T} \sum_{t=1}^{\infty} E(t) = \lim_{T \to \infty} \frac{1}{T} \sum_{t=1}^{\infty} \big|X(t)-\hat{X}(t)\big|.
\end{align}
\begin{lemma}
		\label{theorem_InstError}
For an $N$-state DTMC information source, the time-averaged reconstruction error under the RS policy is given by:
        \begin{align}
\label{PE_RS_lemma}
\bar{E}^{\text{RS}}=\frac{p(N^{2}-1)(1-p^{r}_{\alpha}p_{\text{s}})}{3\big[p^{r}_{\alpha}p_{\text{s}}+Np(1-p^{r}_{\alpha}p_{\text{s}})\big]}.
        \end{align}
Furthermore, under the SARS policy, the time-averaged reconstruction error is expressed as:
\begin{align}
\label{PE_SA_lemma}
\bar{E}^{\text{SARS}}=\frac{p(1-q_{\alpha_{1}}p_{s})(N^{2}-1)}{3\big[Np-(N-1)pq_{\alpha_{1}}p_{s}+(1-p)q_{\alpha_{2}}p_{s}\big]}.
        \end{align}
Moreover, the time-averaged reconstruction error for the
CARS policy can be calculated as:
\begin{align}
    \label{PE_CA_lemma}
    \bar{E}^{\text{CARS}}=\frac{(N^{2}-1)(1-p^{c}_{\alpha}p_{s})}{3\big[N-p^{c}_{\alpha}p_{s}\big]}.
\end{align}
	\end{lemma}
	\begin{IEEEproof} 
		See Appendix \ref{Appendix_LemmaPij_Et}.
	\end{IEEEproof}
    \subsection{Age of Incorrect Information (AoII)}
  The AoII measures the time that has passed since the transmitted information became incorrect or outdated \cite{maatouk2020age}. Let $\text{AoII}(t)$ denote the AoII at time slot $t$. When the system is in a synced state at time $t$, $\text{AoII}(t) = 0$; otherwise, $\text{AoII}(t) = i \geqslant 1$ indicates that the system has been in an erroneous state from time slot $t - i + 1$ to $t$. The evolution of AoII is defined as follows:
    \begin{align}
		\label{AoII_definition}
		\text{AoII}\big(t\big) =
		\begin{cases}
			\text{AoII}(t-1)+1, &X(t)\neq \hat{X}(t),\\
			0, &\text{otherwise}.
		\end{cases}
	\end{align}
For the AoII metric, which captures \emph{the timeliness}, \emph{inaccuracy}, and \emph{the significance of the information} at the receiver, we additionally introduce a threshold-aware randomized stationary (TARS) joint sampling and transmission policy, in addition to the policies presented in Section~\ref{SystemModel}. Under this policy, when $\text{AoII}(t) \geqslant n$, the sampler always performs sampling; when $\text{AoII}(t) = n - 1$, the sampler performs sampling with probability $p^{th}_{\alpha}$; otherwise, the sampler remains idle. 
        \begin{lemma}
		\label{theorem_ConsecutiveError}
		For an $N$-state DTMC information source, the average AoII under the TARS policy is given by:
    \begin{align}
        \label{AvgCE_Threshold}
        \overline{\text{AoII}} = \frac{F(p^{th}_{\alpha},n)}{G(p^{th}_{\alpha},n)},
    \end{align}
    where $F(p^{th}_{\alpha},n)$ and $G(p^{th}_{\alpha},n)$ are given by:
    \begin{align}
    \label{FG_threhold}
        F(p^{th}_{\alpha},n)&\!=\!(N-1)\Big[(1-p)\big(p\!+\!(1-p)p_{s}\big)^{2}\!+\!(1-p)^{n}p_{s}\notag\\
        &\!\times\!\Big(3p^{2}\!\!-\!\!2p\!-\!np^{2}\!-\!p^{3}\!+\!np^{3}(1\!-\!p^{th}_{\alpha})\!-\!p^{2}p^{th}_{\alpha}\!+\!p^{3}p^{th}_{\alpha}\notag\\
        &\!+\!(1-p)p_{s}\big(1\!+\!(n-2)p\!-\!(n-1)p^{2}(1-p^{th}_{\alpha})\big)\Big)\Big],\notag\\
        G(p^{th}_{\alpha},n)&\!=\!p\big(p_{s}+(1-p_{s})p\big)\big[N(1-p)\big(p_{s}+(1-p_{s})p\big)\notag\\
        &-(N-1)(1-p)^{n}\big(1-(1-p^{th}_{\alpha})p\big)p_{s}\big].
    \end{align}
    Furthermore, under the RS policy, the average AoII can be expressed as follows:
    \begin{align}
        \label{AvgCE_RS}
        \overline{\text{AoII}}= \frac{(N-1)p(1-p^{r}_{\alpha}p_{s})}{\big[p+(1-p)p^{r}_{\alpha}p_{s}\big]\big[p^{r}_{\alpha}p_{s}+Np(1-p^{r}_{\alpha}p_{s})\big]}.
    \end{align}
    Moreover, the average AoII for the SARS policy can be calculated as follows:
       \begin{align}
        \label{AvgCE_SARS}
        \overline{\text{AoII}} = \!\!\frac{(N\!-\!1)p(1\!-\!q_{\alpha_{1}}p_{s})}{\big[p\!+\!(1\!-\!p)q_{\alpha_{2}}p_{s}\big]\big[Np\!+\!(1\!-\!N)pq_{\alpha_{1}}p_{s}\!\!+\!\!(1\!-\!p)q_{\alpha_{2}}p_{s}\big]}.
    \end{align}
   Additionally, for the CARS policy, the average AoII is obtained as follows:
           \begin{align}
        \label{AvgCE_CARS}
       \overline{\text{AoII}} =\frac{(N-1)(1-p^{c}_{\alpha}p_{s})}{p(N-p^{c}_{\alpha}p_{s})\big[1+(N-2)p^{c}_{\alpha}p_{s}\big]}.
    \end{align}
	\end{lemma}
	
	\begin{IEEEproof} 
		See Appendix \ref{Appendix_Lemma_CE}.
	\end{IEEEproof}
    \begin{remark}
 Using Lemma~\ref{theorem_InstError} and the expressions in \eqref{AvgCE_RS}–\eqref{AvgCE_CARS}, the average AoII under the RS policy can be expressed as a function of the time-averaged reconstruction error as follows:
        \begin{align}
            \label{AoII_PE_RS}
            \overline{\text{AoII}} = \frac{3\bar{E}^{\text{RS}}}{(N+1)\big[p+(1-p)p^{r}_{\alpha}p_{s}\big]}.
        \end{align}
  Furthermore, for the SARS policy, the expression in \eqref{AoII_PE_RS} can be written as follows:
        \begin{align}
            \label{AoII_PE_SARS}
            \overline{\text{AoII}} = \frac{3\bar{E}^{\text{SARS}}}{(N+1)\big[p+(1-p)q_{\alpha_{2}}p_{s}\big]}.
        \end{align}
   Moreover, we can write the expression in \eqref{AoII_PE_RS} under the CARS policy as follows:
           \begin{align}
            \label{AoII_PE_CARS}
            \overline{\text{AoII}} = \frac{3\bar{E}^{\text{CARS}}}{(N+1)p\big[1+(N-2)p^{c}_{\alpha}p_{s}\big]}.
        \end{align}
    \end{remark}
    \subsection{Cost of Actions under Uncertainty (CoAU)}
The cost of actions under uncertainty, arising from the receiver's inability to determine whether the reconstructed state accurately reflects the true source state, is a critical challenge in real-time remote monitoring systems, especially in the presence of high error rates or costly actuation actions. We have presented several joint sampling and transmission policies that efficiently utilize sampling and transmission resources to reduce system error. To further improve system performance, we now focus on the cost of actions under uncertainty to prevent the actuator from making incorrect decisions based on the reconstructed state at the receiver. To this end, we propose a \emph{randomized actuation} policy in which the actuator takes actions \emph{probabilistically} at each time slot $t$: if the receiver successfully decodes the transmitted information, an action is taken with probability $p_{c_1}$; otherwise, an action is taken with probability $p_{c_2}$. In this context, the system may be synced while the receiver is unable to correctly infer that the reconstructed state matches the true state of the source. This mismatch can arise when the system is in a synced state, but the sampling and transmission policy either prevents the transmitter from sending an update or the transmitted update fails due to channel errors, leaving the receiver uncertain about the system’s synced state. Therefore, an appropriate selection of the actuation probabilities $p_{c_{1}}$ and $p_{c_{2}}$ is essential to prevent incorrect actuation decisions. To capture the number of consecutive time slots during which the actuator performs incorrect actions, we define a cost function called the \emph{Cost of Actions under Uncertainty} (CoAU). \emph{The objective of this cost function is to ensure that the actuator takes correct actions while preventing incorrect actions when the receiver is uncertain about the correctness of the estimated state of the source}. Let $\Delta(t)$ denote the CoAU at time slot $t$. If $\Delta(t)=0$, the actuator performs a correct action at time slot $t$, whereas $\Delta(t)=i\geqslant 1$ indicates that from time slot $t-i+1$ to $t$, the actuator has performed incorrect actions consecutively. According to this metric, $\Delta(t)$ increases by one when the system is synced, i.e., $X(t)=\hat{X}(t)$, but the actuator does not perform an action, $c(t)=0$, or when the system is in an erroneous state, i.e., $X(t)\neq\hat{X}(t)$, but the actuator performs an action, $c(t)=1$. In all other cases, $\Delta(t)$ drops to $0$. We describe the evolution of the CoAU as follows:
    \begin{align}
    \label{CoIA}
        &\Delta(t)=\notag\\
        &\!\begin{cases}
            0,&\!\!\parbox[t]{9cm}{{$\big\{\alpha(t)\!=\!1, h(t)\!=\!1,X(t)\!=\!\hat{X}(t),c(t)\!=\!1\big\},\\ 
            \text{or} \big\{\alpha(t)\!=\!1, h(t)\!=\!0,X(t)\!=\!\hat{X}(t),c(t)\!=\!1\big\},\\
            \text{or} \big\{\alpha(t)\!=\!0,X(t)\!=\!\hat{X}(t),c(t)\!=\!1\big\},\\
            \text{or} \big\{X(t)\neq \hat{X}(t),c(t)=0$\big\},}}\\
            \Delta(t\!-\!1)\!+\!1,&\!\!\parbox[t]{9cm}{{$\big\{X(t)=\hat{X}(t),c(t)=0\big\},\\ \text{or} \{X(t)\neq \hat{X}(t),c(t)=1$\big\}.}}
        \end{cases}
    \end{align}
    \begin{lemma}
		\label{theorem_CoIA}
       When the information source is sampled using the RS policy, the probability $P^{\text{RS}}_{\Delta_{0}}=\mathbb{P}\big[\Delta(t)=0\big]$ is given by:
            \begin{align}
 \label{Delta_RS}
P^{\text{RS}}_{\Delta_{0}}&=\frac{(N-1)p(1-p^{r}_{\alpha}p_{s})}{Np+(1-Np)p^{r}_{\alpha}p_{s}}+p^{r}_{\alpha}p_{s}p_{c_{1}}\notag\\
           &+\frac{(1-p^{r}_{\alpha}p_{s})\big[2p-Np+(1-Np)p^{r}_{\alpha}p_{s}\big]p_{c_{2}}}{Np+(1-Np)p^{r}_{\alpha}p_{s}}.
        \end{align}
	\begin{IEEEproof} 
		See Appendix \ref{Appendix_ControlIncorrectAction}.
	\end{IEEEproof}
    	\end{lemma}
        Similar to Lemma \ref{theorem_CoIA}, one can obtain the probability $\mathbb{P}\big[\Delta(t)=0\big]$ when the source is sampled under the SARS and CARS policies as follows:
    \begin{align}
            P^{\text{SARS}}_{\Delta_{0}}&= \frac{(N-1)p(1-q_{\alpha_{1}}p_{s})}{Np+(1-N)pq_{\alpha_{1}}p_{s}+(1-p)q_{\alpha_{2}}p_{s}}\notag\\
            &+\frac{(N-1)pp_{s}\big[p(q_{\alpha_{1}}-q_{\alpha_{2}})+q_{\alpha_{2}}\big]p_{c_{1}}}{Np+(1-N)pq_{\alpha_{1}}p_{s}+(1-p)q_{\alpha_{2}}p_{s}}\notag\\
            &+\frac{p\big[2-N+(N-1)(1-p)q_{\alpha_{1}}p_{s}\big]p_{c_{2}}}{Np+(1-N)pq_{\alpha_{1}}p_{s}+(1-p)q_{\alpha_{2}}p_{s}}\notag\\
            &+\frac{(1-p)\big(1+(1-N)p\big)q_{\alpha_{2}}p_{s}p_{c_{2}}}{Np+(1-N)pq_{\alpha_{1}}p_{s}+(1-p)q_{\alpha_{2}}p_{s}},\label{PD0_SARS}\\
P^{\text{CARS}}_{\Delta_{0}}&=\frac{(N-1)(1-p^{c}_{\alpha}p_{s})}{N-p^{c}_{\alpha}p_{s}}+(N-1)pp^{c}_{\alpha}p_{s}p_{c_{1}}\notag\\
            &+\bigg[3-2N-(N-1)pp^{c}_{\alpha}p_{s}+\frac{2(N-1)^{2}}{N-p^{c}_{\alpha}p_{s}}\bigg]p_{c_{2}}.\label{PD0_CARS}
        \end{align}

    \section{Constrained Optimization Problems}
\par In this section, we formulate and solve three constrained optimization problems. In the first problem, presented in Section \ref{Optimizationproblem_AvgE}, we aim to find the optimal RS, SARS, and CARS sampling policies that minimize the time-averaged reconstruction error under a constraint on the time-averaged sampling frequency. In the second problem, presented in Section \ref{Optimizationproblem_AvgCE}, the objective is to find the optimal RS, CARS, SARS, and TARS sampling policies that minimize the average AoII, while ensuring that the time-averaged sampling frequency remains below a given threshold. In the third problem, presented in Section \ref{OptimizationProblem_CoIA}, the goal is to determine the optimal actuation probabilities that maximize the probability of correct actions at the actuator, subject to a constraint on the time-averaged actuation frequency.
    \subsection{Minimizing the time-averaged reconstruction error}
    \label{Optimizationproblem_AvgE}
    The objective of this optimization problem is to determine the optimal sampling probabilities $p^{r}_{\alpha}$, $p^{c}_{\alpha}$, $q_{\alpha_{1}}$, and $q_{\alpha_{2}}$ for the RS, CARS, and SARS policies, respectively, in order to minimize the time-averaged reconstruction error. We assume that $0<\eta\leqslant 1$ is a prescribed threshold on the time-averaged sampling frequency. The following optimization problem is then formulated for the RS policy.
     \begin{subequations}
		\label{OptProb_AvgE_RS}
		\begin{align}
			&\underset{p^{r}_{\alpha}}{\text{minimize}}\hspace{0.3cm}{\frac{p(N^{2}-1)(1-p^{r}_{\alpha}p_{\text{s}})}{3\big[p^{r}_{\alpha}p_{\text{s}}+Np(1-p^{r}_{\alpha}p_{\text{s}})\big]}}\label{OptProb_AvgE_RS_Obj}\\
			&\text{subject to}\hspace{0.2cm} \lim_{T \to \infty}\frac{1}{T}\mathbb{E}\Bigg[\sum_{t=1}^{T}\mathbbm{1}\{\alpha(t)=1\}\Bigg]\leqslant\eta,\label{OptProb_AvgE_RS_Const}\\
            &\hspace{1.8cm} 0\leqslant p^{r}_{\alpha} \leqslant 1,
		\end{align}
	\end{subequations}
    where $\mathbbm{1}$ denotes the indicator function and the constraint in \eqref{OptProb_AvgE_RS_Const} represents the time-averaged sampling frequency, which can be simplified as:
    \begin{align}
        \label{SamplingConst_RS}
         \lim_{T \to \infty}\frac{1}{T}\mathbb{E}\Bigg[\sum_{t=1}^{T} \mathbbm{1}\{\alpha(t)=1\}\Bigg]=  p^{r}_{\alpha}.
    \end{align}
Using \eqref{SamplingConst_RS}, the optimization problem in \eqref{OptProb_AvgE_RS} is simplified as:
	 \begin{subequations}
		\label{OptProb_AvgE_RS2}
		\begin{align}
			&\underset{p^{r}_{\alpha}}{\text{minimize}}\hspace{0.3cm}{\frac{p(N^{2}-1)(1-p^{r}_{\alpha}p_{\text{s}})}{3\big[p^{r}_{\alpha}p_{\text{s}}+Np(1-p^{r}_{\alpha}p_{\text{s}})\big]}}\label{OptProb_AvgE_RS_Obj2}\\
			&\text{subject to}\hspace{0.2cm} 0\leqslant p^{r}_{\alpha}\leqslant \eta. \label{OptProb_AvgE_RS_Const2}
		\end{align}
	\end{subequations}
To solve this optimization problem, we first note that the objective function in \eqref{OptProb_AvgE_RS_Obj2} decreases with $p^{r}_{\alpha}$. In other words, the objective function attains its minimum value when $p^{r}_{\alpha}$ is at its maximum. From the constraint in \eqref{OptProb_AvgE_RS_Const2}, the optimal sampling probability is $\eta$. Using \eqref{PE_SA_lemma} for the SARS policy, we formulate the corresponding optimization problem in \eqref{OptProb_AvgE_RS} as follows:
    \begin{subequations}
		\label{OptProb_AvgE_SARSC}
		\begin{align}
			&\underset{q_{\alpha_{1}},q_{\alpha_{2}}}{\text{minimize}}\hspace{0.3cm}{\frac{p(1-q_{\alpha_{1}}p_{s})(N^{2}-1)}{3\big[Np-(N-1)pq_{\alpha_{1}}p_{s}+(1-p)q_{\alpha_{2}}p_{s}\big]}}\label{OptProb_AvgE_SARS_Obj}\\
			&\text{subject to}\hspace{0.2cm} \lim_{T \to \infty}\frac{1}{T}\mathbb{E}\Bigg[\sum_{t=1}^{T} \mathbbm{1}\{\alpha(t)=1\}\Bigg]\leqslant\eta,\label{OptProb_AvgE_SARS_Const}\\
            &\hspace{1.8cm} q_{\alpha_{1}},q_{\alpha_{2}}\in[0,1], q_{\alpha_{2}}\geqslant q_{\alpha_{1}}.
		\end{align}
	\end{subequations}
    \begin{lemma}
	\label{theorem_Optimization_AvgE_SARS}
For an $N$-state DTMC information source, the time-averaged sampling frequency under the SARS policy can be expressed as:
            \begin{align}
 &\lim_{T \to \infty}\frac{1}{T}\mathbb{E}\Bigg[\sum_{t=1}^{T} \mathbbm{1}\{\alpha(t)=1\}\Bigg]\notag\\
 &=\frac{ p(N-1)\big[p(q_{\alpha_{1}}-q_{\alpha_{2}})+q_{\alpha_{2}}\big]}{Np-(N-1)pq_{\alpha_{1}}p_{s}+(1-p)q_{\alpha_{2}}p_{s}}.
        \end{align}
	\end{lemma}
    \begin{IEEEproof} 
		See Appendix \ref{AvgSamplinCost_Proof}.
	\end{IEEEproof}
    Using Lemma \ref{theorem_Optimization_AvgE_SARS}, the optimization problem in \eqref{OptProb_AvgE_SARSC} is simplified as:
     \begin{subequations}
		\label{OptProb_AvgE_SARSC2}
		\begin{align}
			&\underset{q_{\alpha_{1}},q_{\alpha_{2}}}{\text{minimize}}\hspace{0.3cm}{\frac{p(1-q_{\alpha_{1}}p_{s})(N^{2}-1)}{3\big[Np-(N-1)pq_{\alpha_{1}}p_{s}+(1-p)q_{\alpha_{2}}p_{s}\big]}}\label{OptProb_AvgE_SARS_Obj2}\\
			&\text{subject to}\hspace{0.2cm} \frac{ p(N-1)\big[p(q_{\alpha_{1}}-q_{\alpha_{2}})+q_{\alpha_{2}}\big]}{Np-(N-1)pq_{\alpha_{1}}p_{s}+(1-p)q_{\alpha_{2}}p_{s}}\leqslant \eta,\label{OptProb_AvgE_SARS_Const2}\\&\hspace{1.8cm} q_{\alpha_{1}},q_{\alpha_{2}}\in[0,1], q_{\alpha_{2}}\geqslant q_{\alpha_{1}}.
		\end{align}
	\end{subequations}
    It can be shown that the objective function in \eqref{OptProb_AvgE_SARS_Obj2} is a decreasing function of $q_{\alpha_{1}}$ and $q_{\alpha_{2}}$. In other words, the objective function reaches its minimum value when both $q_{\alpha_{1}}$ and $q_{\alpha_{2}}$ are at their maximum. Now, using \eqref{OptProb_AvgE_SARS_Const2}, and considering that $0\leqslant q_{\alpha_{1}}\leqslant 1$, $0\leqslant q_{\alpha_{2}}\leqslant 1$, and $q_{\alpha_{2}}\geqslant q_{\alpha_{1}}$, the optimal sampling probabilities $q^{*}_{\alpha_{1}}$ and $q^{*}_{\alpha_{2}}$ can be obtained as follows: 
    \begin{itemize}

        \item When $\eta p_{s}\geqslant p(N-1)$: The optimal sampling probabilities are $q^{*}_{\alpha_{1}}=\min\bigg\{1,\frac{(1-p)\big(\eta p_{s}-p(N-1)\big)+\eta Np}{\eta (N-1)pp_{s}+p^2(N-1)}\bigg\}$ and $q^{*}_{\alpha_{2}}=1$.
        \item When $\eta p_{s}< p(N-1)$:  If $\eta N p\leqslant(1-p)(p(N-1)-\eta p_{s})$, the optimal sampling probabilities are $q^{*}_{\alpha_{1}}=0$ and $q^{*}_{\alpha_{2}}=\min\Big\{1,\frac{\eta Np}{(1-p)(p(N-1)-\eta p_{s})}\Big\}$. Otherwise, if $\eta N p>(1-p)(p(N-1)-\eta p_{s})$, the optimal sampling probabilities are $q^{*}_{\alpha_{1}}=\min\Big\{1,\frac{\eta Np-(1-p)\big(p(N-1)-\eta p_{s}\big)}{\eta (N-1)pp_{s}+p^{2}(N-1)}\Big\}$ and  $q^{*}_{\alpha_{2}}=1$.
    \end{itemize}
    Using \eqref{PE_CA_lemma}, we can formulate the corresponding optimization problem for the CARS policy as follows:
    \begin{subequations}
		\label{OptProb_AvgE_CARSC}
		\begin{align}
			&\underset{p^{c}_{\alpha}}{\text{minimize}}\hspace{0.3cm}\frac{(N^{2}-1)(1-p^{c}_{\alpha}p_{s})}{3\big[N-p^{c}_{\alpha}p_{s}\big]}\label{OptProb_AvgE_CARS_Obj}\\
			&\text{subject to}\hspace{0.2cm} \lim_{T \to \infty}\frac{1}{T}\mathbb{E}\Bigg[\sum_{t=1}^{T} \mathbbm{1}\{\alpha(t)=1\}\Bigg]\leqslant\eta,\label{OptProb_AvgE_CARS_Const}\\
             &\hspace{1.8cm} 0\leqslant p^{c}_{\alpha} \leqslant 1,
		\end{align}
	\end{subequations}
    \begin{lemma}
	\label{theorem_Optimization_AvgE_CARS}
For an $N$-state DTMC information source, the time-averaged sampling frequency under the CARS policy can be expressed as:
            \begin{align}
 \lim_{T \to \infty}\frac{1}{T}\mathbb{E}\Bigg[\sum_{t=1}^{T} \mathbbm{1}\{\alpha(t)=1\}\Bigg]= (N-1)pp^{c}_{\alpha}.
        \end{align}
	\end{lemma}
     \begin{IEEEproof} 
		See Appendix \ref{AvgSamplinCost_Proof}.
	\end{IEEEproof}
    Using Lemma \ref{theorem_Optimization_AvgE_CARS}, the optimization problem can be simplified as:
    \begin{subequations}
		\label{OptProb_AvgE_CARSC2}
		\begin{align}
			&\underset{p^{c}_{\alpha}}{\text{minimize}}\hspace{0.3cm}\frac{(N^{2}-1)(1-p^{c}_{\alpha}p_{s})}{3\big[N-p^{c}_{\alpha}p_{s}\big]}\label{OptProb_AvgE_CARS_Obj2}\\
			&\text{subject to}\hspace{0.2cm} (N-1)pp^{c}_{\alpha}\leqslant \eta,\label{OptProb_AvgE_CARS_Const2}\\
             &\hspace{1.8cm} 0\leqslant p^{c}_{\alpha} \leqslant 1,
		\end{align}
	\end{subequations}
    We note that the objective function in \eqref{OptProb_AvgE_CARS_Obj2} is a decreasing function of $p^{c}_{\alpha}$. In other words, to minimize the objective function, we must find the maximum value of $p^{c}_{\alpha}$ that satisfies the constraint in \eqref{OptProb_AvgE_CARS_Const2}. Based on the constraint given in \eqref{OptProb_AvgE_CARS_Const2}, the maximum feasible sampling probability $p^{c}_{\alpha}$ is $\min\Big\{1,\frac{\eta}{(N-1)p}\Big\}$.

\subsection{Minimizing the average AoII}
    \label{Optimizationproblem_AvgCE}
    \par The objective of this optimization problem is to determine the optimal TARS sampling policy and the corresponding sampling probabilities $p^{r}_{\alpha}$, $p^{c}_{\alpha}$, $q_{\alpha_{1}}$, and $q_{\alpha_{2}}$ for the RS, SARS, and CARS policies, respectively. The aim is to minimize the average AoII while satisfying a constraint on the time-averaged sampling frequency. Here, the parameter $0<\eta\leqslant 1$ denotes the prescribed threshold for the time-averaged sampling frequency. Focusing first on the TARS sampling policy, we formulate the optimization problem as follows: 
     \begin{subequations}
		\label{Optthreshold_CE}
		\begin{align}
			&\underset{p^{th}_{\alpha},n}{\text{minimize}}\hspace{0.1cm}\frac{F(p^{th}_{\alpha},n)}{G(p^{th}_{\alpha},n)}\label{Optthreshold_CE_Obj}\\
			&\text{subject to}\hspace{0.2cm} \lim_{T \to \infty}\frac{1}{T}\mathbb{E}\Bigg[\sum_{t=1}^{T}\mathbbm{1}\{\alpha(t)=1\}\Bigg]\leqslant\eta,\label{Optthreshold_CE_Const}\\
            &\hspace{1.8cm} 0\leqslant p^{th}_{\alpha} \leqslant 1, n\in\mathbb{N},
		\end{align}
	\end{subequations}
   where $F(p^{th}_{\alpha},n)$ and $G(p^{th}_{\alpha},n)$ were given in \eqref{FG_threhold}.
   \begin{lemma}
	\label{theorem_Optimization_AvgCost_THR}
For an $N$-state DTMC information source, the time-averaged sampling frequency under the TARS policy can be expressed as:
 \begin{align}
 \label{AvgCost_THR_Lemma}
         &\lim_{T \to \infty}\frac{1}{T}\mathbb{E}\Bigg[\sum_{t=1}^{T} \mathbbm{1}\{\alpha(t)=1\}\Bigg]\notag\\
         &\!=\!
         \begin{cases}
             \frac{\big[p(N-1)+\big(p-(Np-1)p_{s}\big)p^{th}_{\alpha}\big]}{Np+(1-p)p_{s}-(N-1)pp^{th}_{\alpha}p_{s}},\hspace{1.7cm} &n=1,\\
             \frac{(N-1)p(1-p)^{n-1}\big(1-(1-p^{th}_{\alpha})p\big)}{N(1-p)\big(p+(1-p)p_{s}\big)-(N-1)(1-p)^{n}\big(1-(1-p^{th}_{\alpha})p\big)p_{s}},&n\geqslant 2.
         \end{cases}
    \end{align}
	\end{lemma}
     \begin{IEEEproof} 
		See Appendix \ref{AvgSamplinCost_Proof}.
	\end{IEEEproof}
    The optimal value of $n$, denoted by $n^{*}$, is the minimum value of $n$ that satisfies the constraint in \eqref{Optthreshold_CE_Const}. To determine $n^{*}$, we first assume $p^{th}_{\alpha}=0$ and find the minimum value of $n$ that satisfies this constraint. Then, by substituting this optimal $n$ into \eqref{Optthreshold_CE_Const}, we obtain the optimal sampling probability that satisfies the optimization constraint. Using \eqref{AvgCost_THR_Lemma}, if $\frac{p(N-1)}{Np+(1-p)p_{s}}\leqslant \eta$, the optimal value of $n$ is $n^{*}=1$ and the optimal sampling probability is $p^{th^{*}}_{\alpha} = \min\bigg\{1,\frac{\eta\big(Np+(1-p)p_{s}\big)-p(N-1)}{p+\eta(N-1)pp_{s}-(Np-1)p_{s}}\bigg\}$. Otherwise, the optimal values of $n$ and $p^{th}_{\alpha}$  that satisfy the constraint of the optimization problem are given as follows:
    \begin{align}
        \label{optimal_n_pa}
n^{*}&=\left\lceil\frac{\log\Big(\frac{N(1-p)(p+(1-p)p_{s})\eta}{(N-1)p+(N-1)(1-p)p_{s}\eta}\Big)}{\log(1-p)}\right\rceil,\notag\\
p^{th^{*}}_{\alpha}\!\!\!&=\!\max\!\Bigg\{\!0, \frac{1-p}{p}\Bigg[\frac{\eta N (1\!-\!p)\big(p_{s}\!+\!(1\!-\!p_{s})p\big)}{(N\!-\!1)(1\!-\!p)^{n^{*}}\big(p\!+\!(1\!-\!p)\eta p_{s}\big)}\!-\!1\Bigg]\!\Bigg\}.
    \end{align}
Now, using \eqref{AvgCE_RS}, the optimization problem for the RS policy is formulated as follows:
\begin{subequations}
		\label{OptProb_AvgCE_RS}
		\begin{align}
			&\underset{p^{r}_{\alpha}}{\text{minimize}}\hspace{0.2cm}{\frac{(N\!-\!1)p(1\!-\!p^{r}_{\alpha}p_{s})}{\big[p\!+\!(1\!-\!p)p^{r}_{\alpha}p_{s}\big]\big[p^{r}_{\alpha}p_{s}\!+\!Np(1\!-\!p^{r}_{\alpha}p_{s})\big]}}\label{OptProb_AvgCE_RS_Obj}\\
			&\text{subject to}\hspace{0.2cm} \lim_{T \to \infty}\frac{1}{T}\mathbb{E}\Bigg[\sum_{t=1}^{T} \mathbbm{1}\{\alpha(t)=1\}\Bigg]\leqslant\eta,\label{OptProb_AvgCE_RS_Const},\\
            &\hspace{1.8cm} 0\leqslant p^{r}_{\alpha} \leqslant 1.
		\end{align}
	\end{subequations}
    Using \eqref{SamplingConst_RS}, the optimization problem in \eqref{OptProb_AvgCE_RS} is simplified as:
    \begin{subequations}
		\label{OptProb_AvgCE_RS2}
		\begin{align}
			&\underset{p^{r}_{\alpha}}{\text{minimize}}\hspace{0.2cm}{\frac{(N\!-\!1)p(1\!-\!p^{r}_{\alpha}p_{s})}{\big[p\!+\!(1\!-\!p)p^{r}_{\alpha}p_{s}\big]\big[p^{r}_{\alpha}p_{s}\!+\!Np(1\!-\!p^{r}_{\alpha}p_{s})\big]}}\label{OptProb_AvgCE_RS_Obj2}\\
			&\text{subject to}\hspace{0.2cm} 0\leqslant p^{r}_{\alpha}\leqslant \eta.\label{OptProb_AvgCE_RS_Const2}
		\end{align}
	\end{subequations}
    To solve this optimization problem, we note that the objective function in \eqref{OptProb_AvgCE_RS_Obj2} decreases as $p^{r}_{\alpha}$ increases. Therefore, the problem reduces to finding the maximum value of $p^{r}_{\alpha}$ that satisfies the constraint in \eqref{OptProb_AvgCE_RS_Const2}. Using this constraint, the optimal sampling probability is $\eta$. Furthermore, using \eqref{AvgCE_SARS}, we can  formulate the corresponding optimization problem for the SARS policy as follows:
     \begin{subequations}
		\label{OptProb_AvgCE_SARSC}
		\begin{align}
		&\underset{q_{\alpha_{1}},q_{\alpha_{2}}}{\text{minimize}}\hspace{0cm}{\frac{(N\!-\!1)p(1\!-\!q_{\alpha_{1}}p_{s})}{\big[p\!+\!(1\!-\!p)q_{\alpha_{2}}p_{s}\big]\!\big[\!Np\!+\!(1\!\!-\!\!N)pq_{\alpha_{1}}p_{s}\!\!+\!\!(1\!-\!p)q_{\alpha_{2}}p_{s}\!\big]}}\label{OptProb_AvgCE_SARS_Obj}\\
			&\text{subject to}\hspace{0.2cm} \lim_{T \to \infty}\frac{1}{T}\mathbb{E}\Bigg[\sum_{t=1}^{T} \mathbbm{1}\{\alpha(t)=1\}\Bigg]\leqslant\eta,\label{OptProb_AvgCE_SARS_Const}\\
            &\hspace{1.8cm} q_{\alpha_{1}},q_{\alpha_{2}}\in[0,1], q_{\alpha_{2}}\geqslant q_{\alpha_{1}}.
		\end{align}
	\end{subequations}
    Using Lemma~\ref{theorem_Optimization_AvgE_SARS}, the optimization problem in \eqref{OptProb_AvgCE_SARSC} is simplified as:
    \begin{subequations}
		\label{OptProb_AvgCE_SARSC2}
		\begin{align}
		&\underset{q_{\alpha_{1}},q_{\alpha_{2}}}{\text{minimize}}\hspace{0cm}{\frac{(N\!-\!1)p(1\!-\!q_{\alpha_{1}}p_{s})}{\big[p\!+\!(1\!-\!p)q_{\alpha_{2}}p_{s}\big]\!\big[\!Np\!+\!(1\!\!-\!\!N)pq_{\alpha_{1}}p_{s}\!\!+\!\!(1\!-\!p)q_{\alpha_{2}}p_{s}\!\big]}}\label{OptProb_AvgCE_SARS_Obj2}\\
			&\text{subject to}\hspace{0.2cm} \frac{ p(N-1)\big[p(q_{\alpha_{1}}-q_{\alpha_{2}})+q_{\alpha_{2}}\big]}{Np-(N-1)pq_{\alpha_{1}}p_{s}+(1-p)q_{\alpha_{2}}p_{s}}\leqslant \eta,\label{OptProb_AvgCE_SARS_Const2}\\
            &\hspace{1.8cm} q_{\alpha_{1}},q_{\alpha_{2}}\in[0,1], q_{\alpha_{2}}\geqslant q_{\alpha_{1}}.
		\end{align}
	\end{subequations}
   It can be shown that the objective function in \eqref{OptProb_AvgCE_SARS_Obj2} decreases with both $q_{\alpha_{1}}$ and $q_{\alpha_{2}}$. This means that the objective function reaches its minimum when $q_{\alpha_{1}}$ and $q_{\alpha_{2}}$ take their maximum values. Using \eqref{OptProb_AvgCE_SARS_Const2}, and noting that $0\leqslant q_{\alpha_{1}}\leqslant 1$, $0\leqslant q_{\alpha_{2}}\leqslant 1$, and $q_{\alpha_{2}}\geqslant q_{\alpha_{1}}$, the optimal sampling probabilities $q^{*}_{\alpha_{1}}$ and $q^{*}_{\alpha_{2}}$ can be calculated as follows: 
    \begin{itemize}
        \item When $\eta p_{s}\geqslant p(N-1)$: The optimal sampling probabilities are $q^{*}_{\alpha_{1}}=\min\bigg\{1,\frac{(1-p)\big(\eta p_{s}-p(N-1)\big)+\eta Np}{\eta (N-1)pp_{s}+p^2(N-1)}\bigg\}$ and $q^{*}_{\alpha_{2}}=1$.
        \item When $\eta p_{s}< p(N-1)$:  If $\eta N p\leqslant(1-p)\big(p(N-1)-\eta p_{s}\big)$, the optimal sampling probabilities are $q^{*}_{\alpha_{1}}=0$ and $q^{*}_{\alpha_{2}}=\min\Big\{1,\frac{\eta Np}{(1-p)(p(N-1)-\eta p_{s})}\Big\}$. Otherwise, if $\eta N p>(1-p)\big(p(N-1)-\eta p_{s}\big)$, the optimal sampling probabilities are $q^{*}_{\alpha_{1}}=\min\Big\{1,\frac{\eta Np-(1-p)\big(p(N-1)-\eta p_{s}\big)}{\eta (N-1)pp_{s}+p^{2}(N-1)}\Big\}$ and  $q^{*}_{\alpha_{2}}=1$.
    \end{itemize}
We now use \eqref{AvgCE_CARS} to formulate the corresponding optimization problem for the CARS policy as follows:
    \begin{subequations}
		\label{OptProb_AvgCE_CARSC}
		\begin{align}
			&\underset{p^{c}_{\alpha}}{\text{minimize}}\hspace{0.3cm}\frac{(N-1)(1-p^{c}_{\alpha}p_{s})}{p(N-p^{c}_{\alpha}p_{s})\big[1+(N-2)p^{c}_{\alpha}p_{s}\big]}\label{OptProb_AvgCE_CARS_Obj}\\
			&\text{subject to}\hspace{0.2cm} \lim_{T \to \infty}\frac{1}{T}\mathbb{E}\Bigg[\sum_{t=1}^{T}\mathbbm{1}\{\alpha(t)=1\}\Bigg]\leqslant\eta,\label{OptProb_AvgCE_CARS_Const}\\
            &\hspace{1.8cm} 0\leqslant p^{c}_{\alpha} \leqslant 1.
		\end{align}
	\end{subequations}
    Using Lemma~\ref{theorem_Optimization_AvgE_CARS}, the optimization problem in \eqref{OptProb_AvgCE_CARSC} can be simplified as:
     \begin{subequations}
		\label{OptProb_AvgCE_CARSC2}
		\begin{align}
			&\underset{p^{c}_{\alpha}}{\text{minimize}}\hspace{0.3cm}\frac{(N-1)(1-p^{c}_{\alpha}p_{s})}{p(N-p^{c}_{\alpha}p_{s})\big[1+(N-2)p^{c}_{\alpha}p_{s}\big]}\label{OptProb_AvgCE_CARS_Obj2}\\
			&\text{subject to}\hspace{0.2cm} (N-1)pp^{c}_{\alpha}\leqslant \eta,\label{OptProb_AvgCE_CARS_Const2}\\&\hspace{1.8cm} 0\leqslant p^{c}_{\alpha} \leqslant 1.
		\end{align}
	\end{subequations}
   The objective function in \eqref{OptProb_AvgCE_CARS_Obj2} decreases with $p^{c}_{\alpha}$. This means that to minimize the objective function, we need to find the maximum value of $p^{c}_{\alpha}$ that satisfies the constraint in \eqref{OptProb_AvgCE_CARS_Const2}. From this constraint, the maximum feasible sampling probability $p^{c}_{\alpha}$ is $\min\Big\{1, \frac{\eta}{(N-1)p}\Big\}$.
   
\subsection{Maximizing the probability of a correct actuation decision}
    \label{OptimizationProblem_CoIA}
  The objective of this optimization problem is to determine the optimal actuation probabilities $p_{c_1}$ and $p_{c_2}$ when the source is sampled using the RS, SARS, and CARS policies. This is achieved by maximizing the probability that the CoAU is zero, while ensuring that the time-averaged actuation frequency remains below a given threshold. We assume that the parameter $0<\mu\leqslant 1$ represents the threshold for the time-averaged actuation frequency. This optimization problem is formulated as follows:
    \begin{subequations}
		\label{Optimization_problem_PD0_RS}
		\begin{align}
			&\underset{p_{c_{1}},p_{c_{2}}}{\text{maximize}}\hspace{0.3cm}\mathbb{P}\big[\Delta(t)=0\big]\label{Optimization_prob_PD0_RS_objfunc}\\
			&\text{subject to}\hspace{0.2cm} \lim_{T \to \infty}\frac{1}{T}\mathbb{E}\Bigg[\sum_{t=1}^{T} \mathbbm{1}\{c(t)=1\} \Bigg]\leqslant\mu,\label{Optimization_prob_PD0_constraint1}\\
            &\hspace{1.8cm} p_{c_{1}},p_{c_{2}}\in[0,1].
		\end{align}
	\end{subequations}
  \begin{lemma}
	\label{theorem_Optimization_PD0_RS}
When the source is sampled using the RS policy, the time-averaged actuation frequency can be expressed as follows:
\begin{align}
            \label{AvgCost_CoIA_RS}
&\lim_{T \to \infty}\frac{1}{T}\mathbb{E}\Bigg[\sum_{t=1}^{T} \mathbbm{1}\{c(t)\!=\!1\}\Bigg]\!=\!p^{r}_{\alpha}p_{s}p_{c_{1}}\!\!+\!\!(1\!-\!p^{r}_{\alpha}p_{s})p_{c_{2}}.
        \end{align}
	\end{lemma}
     \begin{IEEEproof} 
		See Appendix \ref{Appendix_CoIA_TAvgSamplingCost}.
	\end{IEEEproof}  
Using \eqref{Delta_RS} and Lemma~\ref{theorem_Optimization_PD0_RS}, the optimal actuation probabilities $p^{*}_{c_{1}}$ and $p^{*}_{c_{2}}$ are obtained as follows:
    \begin{itemize}
        \item When $(1-Np)p^{r}_{\alpha}p_{s}\leqslant (N-2)p$: The optimal actuation probabilities are  $p^{*}_{c_{1}}=\min\Big\{1,\frac{\mu}{p^{r}_{\alpha}p_{s}}\Big\}$ and $p^{*}_{c_{2}}=0$.\\
        \item When $(1-Np)p^{r}_{\alpha}p_{s}> (N-2)p$: If $\mu \geqslant p^{r}_{\alpha}p_{s}$, the optimal actuation probabilities are $p^{*}_{c_{1}}=1$ and $p^{*}_{c_{2}}=\min\Big\{1,\frac{\mu-p^{r}_{\alpha}p_{s}}{1-p^{r}_{\alpha}p_{s}}\Big\}$. Otherwise, if $\mu <p^{r}_{\alpha}p_{s}$, $p^{*}_{c_{1}}=\min\Big\{1,\frac{\mu}{p^{r}_{\alpha}p_{s}}\Big\}$ and $p^{*}_{c_{2}}=0$.
    \end{itemize}
    \begin{lemma}
	\label{theorem_Optimization_PD0_SARS}
When the source is sampled using the SARS policy, the time-averaged actuation frequency is given by:
            \begin{align}
            \label{AvgCost_CoIA_SARS}
&\lim_{T \to \infty}\frac{1}{T}\mathbb{E}\Bigg[\sum_{t=1}^{T}\mathbbm{1}\{c(t)=1\}\Bigg]\notag\\
        &=\Bigg(\frac{(N-1)pp_{s}\big[p(q_{\alpha_{1}}-q_{\alpha_{2}})+q_{\alpha_{2}}\big]}{Np+(1-N)pq_{\alpha_{1}}p_{s}+(1-p)q_{\alpha_{2}}p_{s}}\Bigg)p_{c_{1}}\notag\\
        &+\Bigg(\frac{Np + (1 - N)p (1 + p)  q_{\alpha_{1}}p_{s}}{Np+(1-N)pq_{\alpha_{1}}p_{s}+(1-p)q_{\alpha_{2}}p_{s}}\notag\\
        &+\frac{(1-p)\big(1+(1-N)p\big)q_{\alpha_{2}}p_{s}}{Np+(1-N)pq_{\alpha_{1}}p_{s}+(1-p)q_{\alpha_{2}}p_{s}}\Bigg)p_{c_{2}}.
        \end{align}
	\end{lemma}
     \begin{IEEEproof} 
		See Appendix \ref{Appendix_CoIA_TAvgSamplingCost}.
	\end{IEEEproof}  

Using \eqref{PD0_SARS} and Lemma~\ref{theorem_Optimization_PD0_SARS}, the optimal actuation probabilities $p^{*}_{c_{1}}$ and $p^{*}_{c_{2}}$ are calculated as follows:
\begin{itemize}
        \item When $p(N-1)(1-p)q_{\alpha_{1}}p_{s}\leqslant p(N-2)-(1-p)\big(1+(1-N)p\big)q_{\alpha_{2}}p_{s}$: The optimal actuation probabilities are  $p^{*}_{c_{1}}=\min\Big\{1,\frac{\mu f(q_{\alpha_{1}},q_{\alpha_{2}})}{g(q_{\alpha_{1}},q_{\alpha_{2}})} \Big\}$ and $p^{*}_{c_{2}}=0$.
        \item When $p(N-1)(1-p)q_{\alpha_{1}}p_{s}> p(N-2)-(1-p)\big(1+(1-N)p\big)q_{\alpha_{2}}p_{s}$: If $(N-1)pp_{s}\big[p(q_{\alpha_{1}}-q_{\alpha_{2}})+q_{\alpha_{2}}\big]\leqslant \mu \big[Np+(1-N)pq_{\alpha_{1}}p_{s}+(1-p)q_{\alpha_{2}}p_{s}\big]$, the optimal actuation probabilities are $p^{*}_{c_{1}}=1$ and $p^{*}_{c_{2}}=\min\bigg\{1,\frac{\mu f(q_{\alpha_{1}},q_{\alpha_{2}})-g(q_{\alpha_{1}},q_{\alpha_{2}})}{Np+(1-N)p(1+p)q_{\alpha_{1}}p_{s}+(1-p)\big(1+(1-N)p\big)q_{\alpha_{2}}p_{s}}\bigg\}$. Otherwise, $p^{*}_{c_{1}}=\min\Big\{1,\frac{\mu f(q_{\alpha_{1}},q_{\alpha_{2}})}{g(q_{\alpha_{1}},q_{\alpha_{2}})} \Big\}$ and $p^{*}_{c_{2}}=0$. 
    \end{itemize}
    where $f(q_{\alpha_{1}},q_{\alpha_{2}})$ and $g(q_{\alpha_{1}},q_{\alpha_{2}})$ are given by:
    \begin{align}
        \label{fg}
        f\big(q_{\alpha_{1}},q_{\alpha_{2}}\big)&= Np+(1-N)pq_{\alpha_{1}}p_{s}+(1-p)q_{\alpha_{2}}p_{s},\notag\\
g\big(q_{\alpha_{1}},q_{\alpha_{2}}\big)&=(N-1)pp_{s}\big[p(q_{\alpha_{1}}-q_{\alpha_{2}})+q_{\alpha_{2}}\big].
    \end{align}
        \begin{lemma}
	\label{theorem_Optimization_PD0_CARS}
When the source is sampled using the CARS policy, the time-averaged actuation frequency can be expressed as:
            \begin{align}
            \label{AvgCost_CoIA_CARS}
&\lim_{T \to \infty}\frac{1}{T}\mathbb{E}\Bigg[\sum_{t=1}^{T} \mathbbm{1}\{c(t)=1\}\Bigg]\notag\\
        &=\big[(N-1)pp^{c}_{\alpha}p_{s}p_{c_{1}}+\big(1+(1-N)pp^{c}_{\alpha}p_{s}\big)p_{c_{2}}\big].
        \end{align}
	\end{lemma}
     \begin{IEEEproof} 
		See Appendix \ref{Appendix_CoIA_TAvgSamplingCost}.
	\end{IEEEproof} 
    Using \eqref{PD0_CARS} and Lemma~\ref{theorem_Optimization_PD0_CARS}, the optimal actuation probabilities $p^{*}_{c_{1}}$ and $p^{*}_{c_{2}}$ are obtained as:
    \begin{itemize}
        \item When $(N-p^{c}_{\alpha}p_{s})\big[2N-3+(N-1)pp^{c}_{\alpha}p_{s}\big]\geqslant2(N-1)^{2}$: The optimal actuation probabilities are $p^{*}_{c_{2}}=0$ and $p^{*}_{c_{1}}=\min\Big\{1,\frac{\mu}{(N-1)pp^{c}_{\alpha}p_{s}}\Big\}$.
        \item When $(N-p^{c}_{\alpha}p_{s})\big[2N-3+(N-1)pp^{c}_{\alpha}p_{s}\big]<2(N-1)^{2}$: If $(N-1)pp^{c}_{\alpha}p_{s}<\mu$, the optimal actuation probabilities are $p^{*}_{c_{1}}=1$ and $p^{*}_{c_{2}}=\min\Big\{1,\frac{\mu-(N-1)pp^{c}_{\alpha}p_{s}}{1+(1-N)pp^{c}_{\alpha}p_{s}}\Big\}$. Otherwise, if $(N-1)pp^{c}_{\alpha}p_{s}\geqslant \mu$, $p^{*}_{c_{1}}=\min\Big\{1,\frac{\mu}{(N-1)pp^{c}_{\alpha}p_{s}}\Big\}$ and $p^{*}_{c_{2}}=0$.
    \end{itemize}
    \section{Numerical Results}
    In this section, we numerically evaluate and discuss the analytical results for a $N = 3$-state Markov chain information source. We assess and compare the average AoII under the optimal RS, CARS, SARS, and TARS policies. We then investigate the probability that the CoAU equals zero under the optimal randomized actuation policy and the always-act policy.
    \par \par Figs.~\ref{MinAvgAoII_q0.1} and \ref{MinAvgAoII_q0.8} show the minimum average AoII under the time-averaged sampling frequency constraint as a function of $\eta$ for selected values of $q$ and the success probability. We observe that for a rapidly changing source (i.e. smaller value of $q$), where the system is required to take more sampling actions to be in a synced state, the optimal SARS policy achieves superior performance compared to the other policies. In fact, the main advantage of the optimal SARS policy is that it maintains a low average AoII while avoiding unnecessary sampling actions. Under the SARS policy, sampling at time slot $t$ is performed \emph{probabilistically} only when the system is in an erroneous state, i.e., when $X(t)\neq \hat{X}(t-1)$.  \emph{In this policy, more sampling actions are performed when the system remains in an erroneous state for several consecutive time slots than when the system was synced in the previous slot but becomes erroneous in the current slot, i.e., $q_{\alpha_{2}}\geqslant q_{\alpha_{1}}$. This adaptive behavior reduces the average AoII and improves overall system performance}.  For example, Fig.~\ref{MinAvgAoII_q0.1ps0.1} illustrates that, for a target average AoII of $1.3$, the optimal RS policy reduces the sampling action from $0.70$ to $0.67$, which corresponds to a $4.2\%$ reduction compared to the optimal CARS policy. The optimal TARS policy further reduces the sampling action to $0.48$, achieving a $31.4\%$ reduction relative to the optimal CARS policy. In contrast, the optimal SARS policy requires $57.1\%$ fewer sampling actions (to $0.3$) than the optimal CARS policy. \emph{This ability to balance sampling frequency is important for improving performance in dynamic systems and makes the SARS policy suitable for resource-constrained environments such as energy-limited wireless networks}. In contrast, when the source remains in the same state with high probability and the allowed time-averaged sampling frequency is small (i.e. larger value of $q$ and smaller value of $\eta$), the optimal TARS policy performs better than the other policies. However, when the sampling frequency budget is large, the optimal SARS policy achieves better performance. This is because, under a strict sampling frequency constraint, the optimal SARS policy samples with very low probability even when the system is in an incorrect state. In this case, the TARS policy becomes more effective since it triggers sampling only when the AoII becomes large enough to significantly affect system performance. 
    \begin{figure}[!t]
		\centering
 
		\subfigure[$p_{{\text{s}}} = 0.1$ ]{\includegraphics[trim=0.5cm 0.05cm 1.1cm 0.6cm,width=0.48\linewidth, clip]{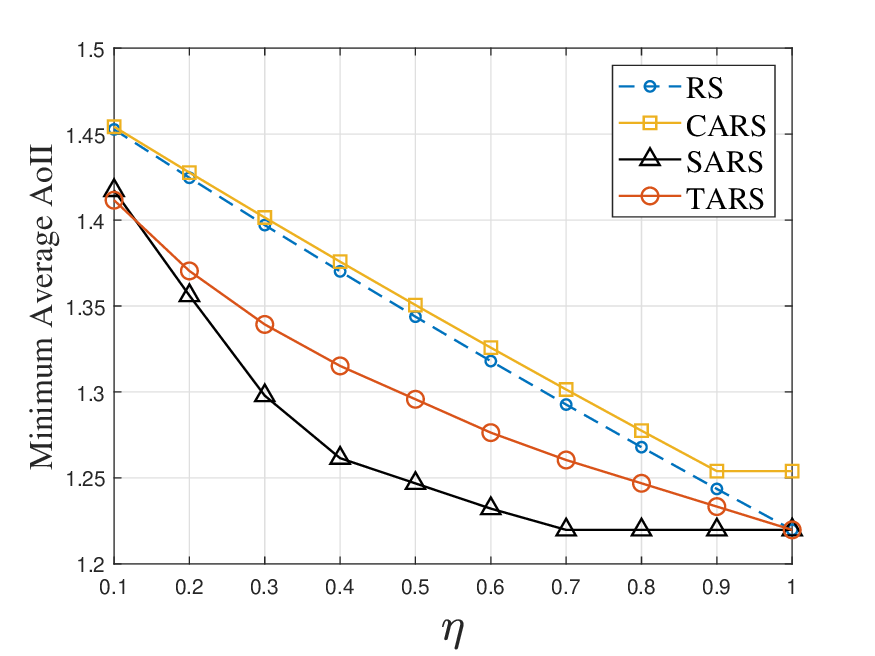}
  \label{MinAvgAoII_q0.1ps0.1}
		}
		\subfigure[$p_{{\text{s}}} = 0.9$]{\centering
			\includegraphics[trim=0.5cm 0.05cm 1.1cm 0.6cm,width=0.48\linewidth, clip]{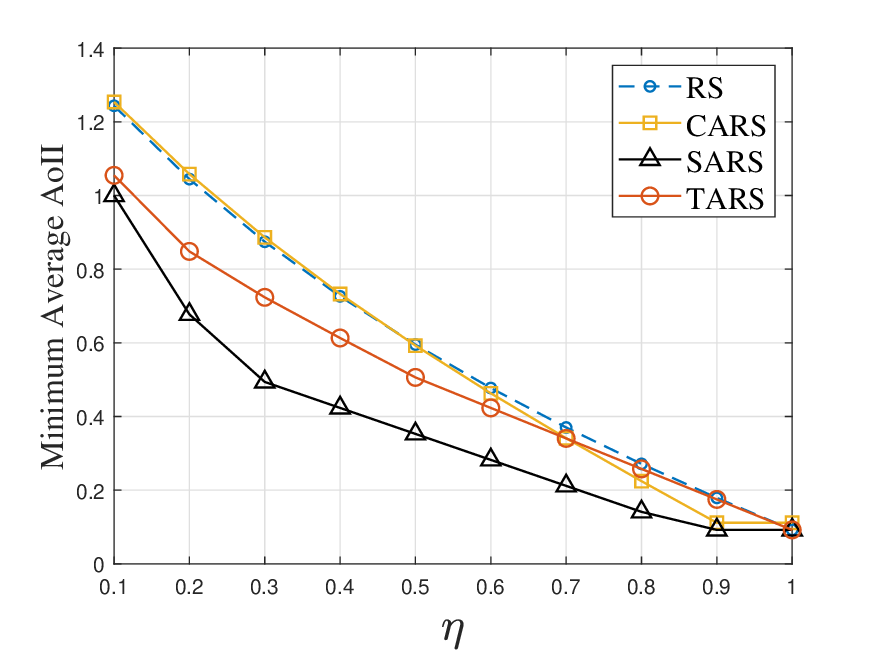}
		}
		\caption{Minimum average AoII as a function of $\eta$ for $q = 0.1$.}
		\label{MinAvgAoII_q0.1}
	\end{figure}
       \begin{figure}[!t]
		\centering
 
		\subfigure[$p_{{\text{s}}} = 0.1$ ]{\includegraphics[trim=0.5cm 0.05cm 1.1cm 0.6cm,width=0.48\linewidth, clip]{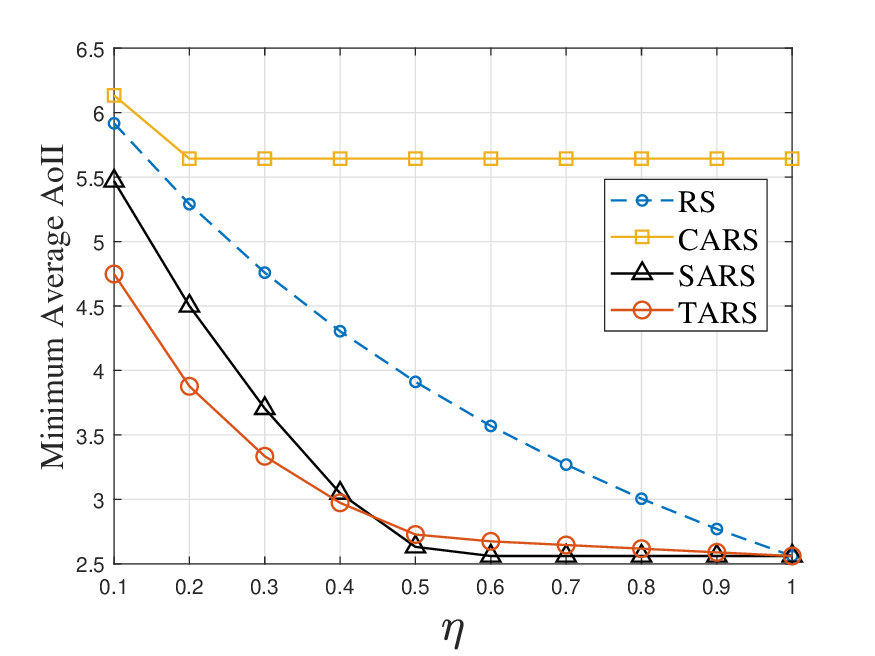}
  \label{MinAvgAoII_q0.8ps0.1}
		}
		\subfigure[$p_{{\text{s}}} = 0.9$]{\centering
			\includegraphics[trim=0.5cm 0.05cm 1.1cm 0.6cm,width=0.48\linewidth, clip]{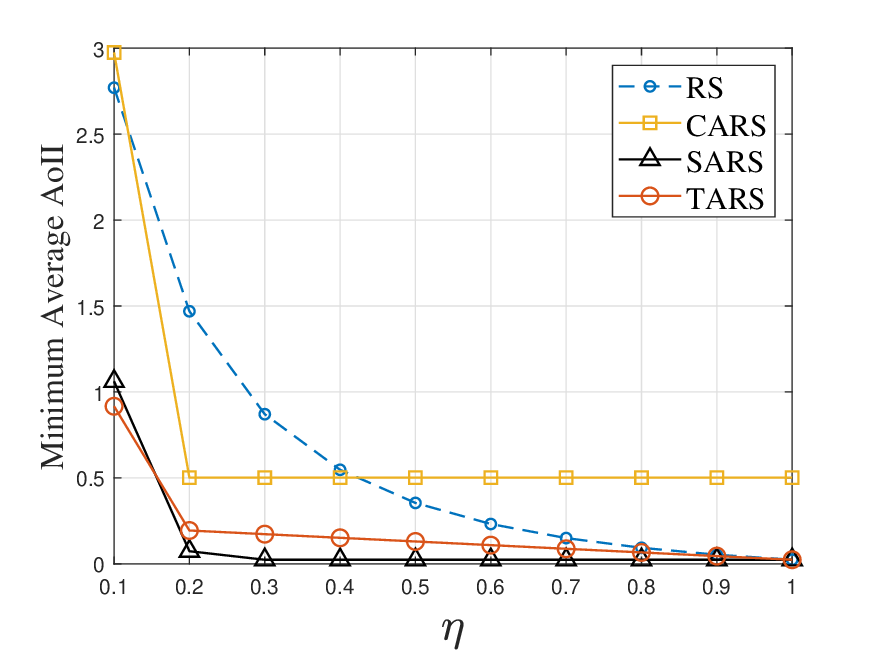}
		}
		\caption{Minimum average AoII as a function of $\eta$ for $q = 0.8$.}
		\label{MinAvgAoII_q0.8}
	\end{figure}
\par Figs.~\ref{TimeAvgSamplingAction_q0.1} and \ref{TimeAvgSamplingAction_q0.8} show the time-averaged sampling frequency as a function of $\eta$ for selected values of $q$ and the success probability. The results indicate that the optimal RS, CARS, SARS, and TARS policies generate samples in a way that ensures the time-averaged sampling frequency does not exceed the threshold $\eta$. Among these policies, the optimal SARS policy can reduce the number of sampling actions by selecting a larger value for $q_{\alpha_{2}}$ and a smaller value for $q_{\alpha_{1}}$, especially when the source changes rapidly. This is because not sampling when the system remains in an erroneous state for several consecutive time slots has a greater detrimental impact on performance than not sampling when the system was synced in the previous slot but becomes erroneous in the current slot. Therefore, it is reasonable to set $q_{\alpha_{2}}$ higher than $q_{\alpha_{1}}$.
    \begin{figure}[!t]
		\centering
 
		\subfigure[$p_{{\text{s}}} = 0.1$ ]{\includegraphics[trim=0.5cm 0.05cm 1.1cm 0.6cm,width=0.48\linewidth, clip]{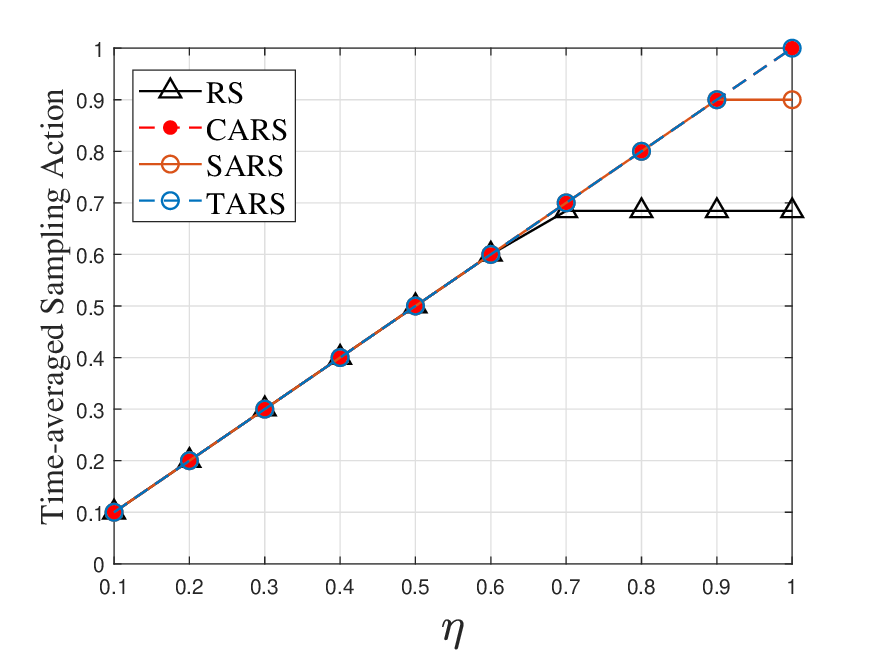}
  \label{TimeAvgSamplingAction_q0.1ps0.1}
		}
		\subfigure[$p_{{\text{s}}} = 0.9$]{\centering
			\includegraphics[trim=0.5cm 0.05cm 1.1cm 0.6cm,width=0.48\linewidth, clip]{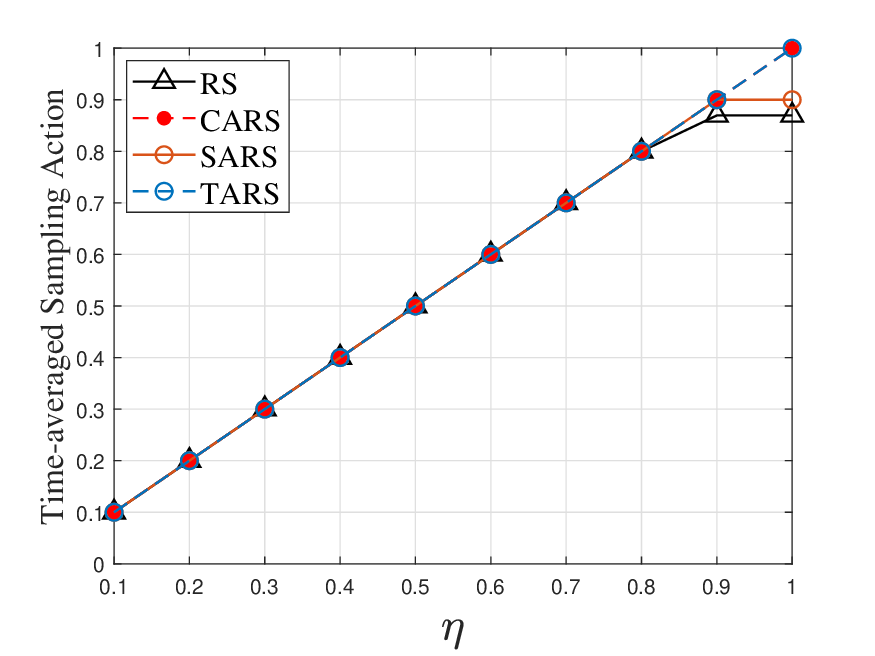}
		}
		\caption{Time-averaged sampling frequency as a function of $\eta$ for $q = 0.1$.}
		\label{TimeAvgSamplingAction_q0.1}
	\end{figure}
       \begin{figure}[!t]
		\centering
 
		\subfigure[$p_{{\text{s}}} = 0.1$ ]{\includegraphics[trim=0.5cm 0.05cm 1.1cm 0.6cm,width=0.48\linewidth, clip]{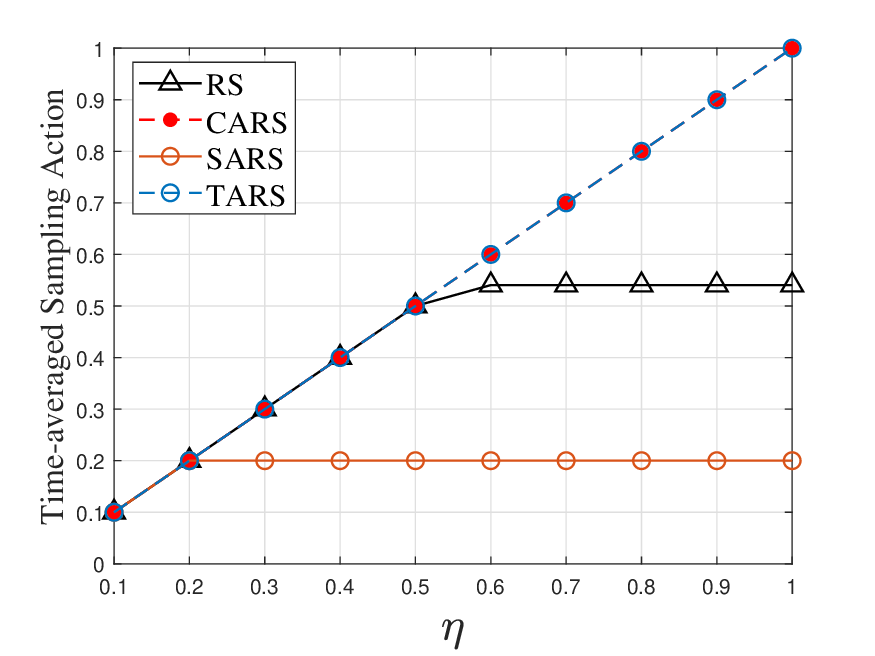}
  \label{TimeAvgSamplingAction_q0.8ps0.1}
		}
		\subfigure[$p_{{\text{s}}} = 0.9$]{\centering
			\includegraphics[trim=0.5cm 0.05cm 1.1cm 0.6cm,width=0.48\linewidth, clip]{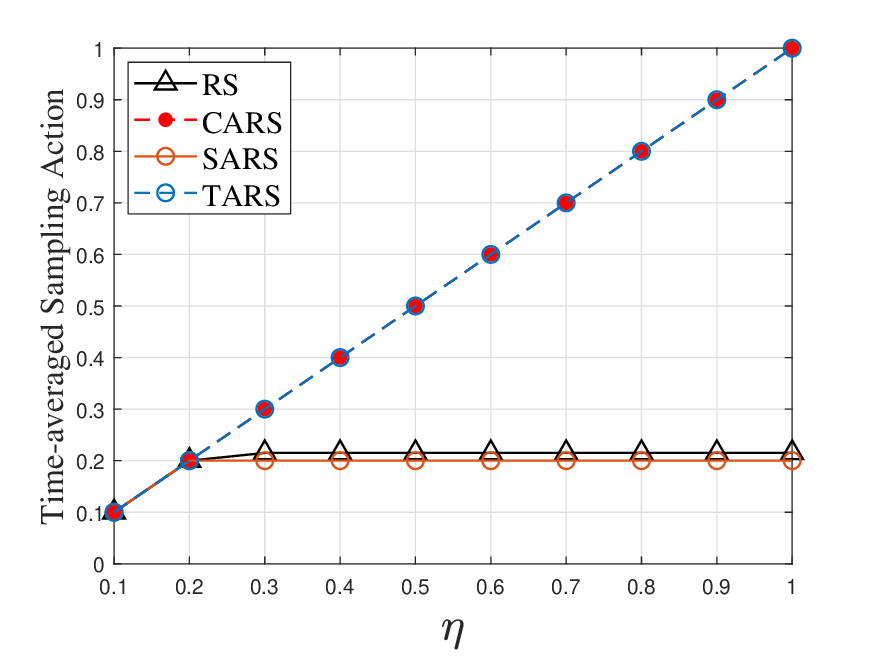}
		}
		\caption{Time-averaged sampling frequency as a function of $\eta$ for $q = 0.8$.}
		\label{TimeAvgSamplingAction_q0.8}
	\end{figure}
     \par Figs.~\ref{PD0RS}, \ref{PD0SARS}, and \ref{PD0CARS} illustrate the probability that the CoAU equals zero under the optimal randomized actuation policy and the always-act policy as a function of $q$, and for selected values of $\eta$ and the success probability. To facilitate a comparison with the always-act policy, we set $\mu=1$. Since the always-act policy performs an actuation in every time slot, its time-averaged actuation frequency is equal to one, making it infeasible for any $\mu<1$. In these figures, $P_{\Delta_{0}}$ under the always-act policy is obtained by setting $p_{c_{1}} = p_{c_{2}} = 1$. For all policies, the sampling and transmission processes use their respective optimal probabilities. The results indicate that when $q$ is small, the probability of taking a correct actuation decision becomes lower. This is because a small $q$ represents a rapidly changing source, which requires more frequent sampling. In this case, the sampling probabilities in the RS, CARS, and SARS policies decrease, making it more likely that the system stays in an erroneous state. By choosing optimal values for $p_{c_{1}}$ and $p_{c_{2}}$, the system can avoid incorrect actuation decisions with higher probability, which increases the number of correct actuation decisions. For example, when $q=0.1$ and $\eta=0.1$, which indicate a rapidly changing source with low sampling and transmission rates, the probability of correct actuation decisions increases from $0.34$ to $0.67$ for $p_{s}=0.1$. This corresponds to a $97\%$ increase in correct actuation decisions compared to always-act actuation. Furthermore, for $p_{s}=0.9$, the probability of correct actuation decisions increases by $86.8\%$ (from $0.38$ to $0.71$) under the optimal RS and CARS policies, while under the optimal SARS policy it increases by $72.5\%$ (from $0.40$ to $0.69$). \emph{It is worth noting that this improvement in correct actuation decisions results from both performing correct actuation decisions and avoiding incorrect actuation decisions when the receiver is uncertain about the correctness of the reconstructed source state}. In addition, increasing $\eta$ leads to a higher number of correct actuation decisions. Since $\eta$ represents the total time-averaged sampling frequency, a larger $\eta$ allows higher sampling and transmission probabilities, which improves the accuracy of the system’s actuation decisions.
       \begin{figure}[!t]
		\centering
 
		\subfigure[$\eta = 0.1$ ]{\includegraphics[trim=0.3cm 0.05cm 1.1cm 0.6cm,width=0.48\linewidth, clip]{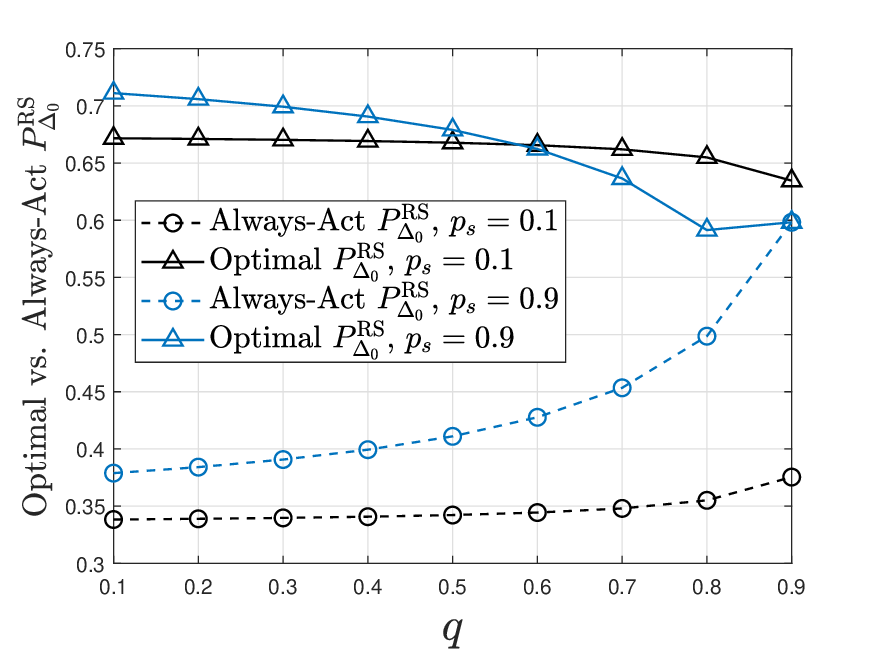}
  \label{PD0RS_eta0.1}
		}
		\subfigure[$\eta = 0.7$]{\centering
			\includegraphics[trim=0.3cm 0.05cm 1.1cm 0.6cm,width=0.48\linewidth, clip]{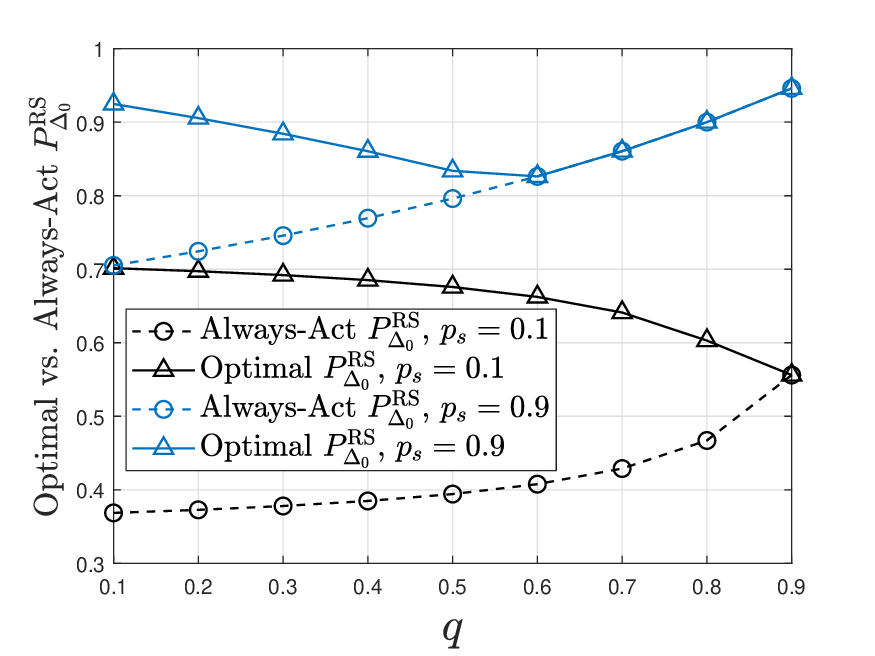}
		}
		\caption{Comparison between the optimal and Always-Act
$P_{\Delta_{0}}$ when the source is sampled using the optimal RS policy as a function of $q$ for $\mu = 1$.}
		\label{PD0RS}
	\end{figure}
            \begin{figure}[!t]
		\centering
 
		\subfigure[$\eta = 0.1$ ]{\includegraphics[trim=0.3cm 0.05cm 1.1cm 0.6cm,width=0.48\linewidth, clip]{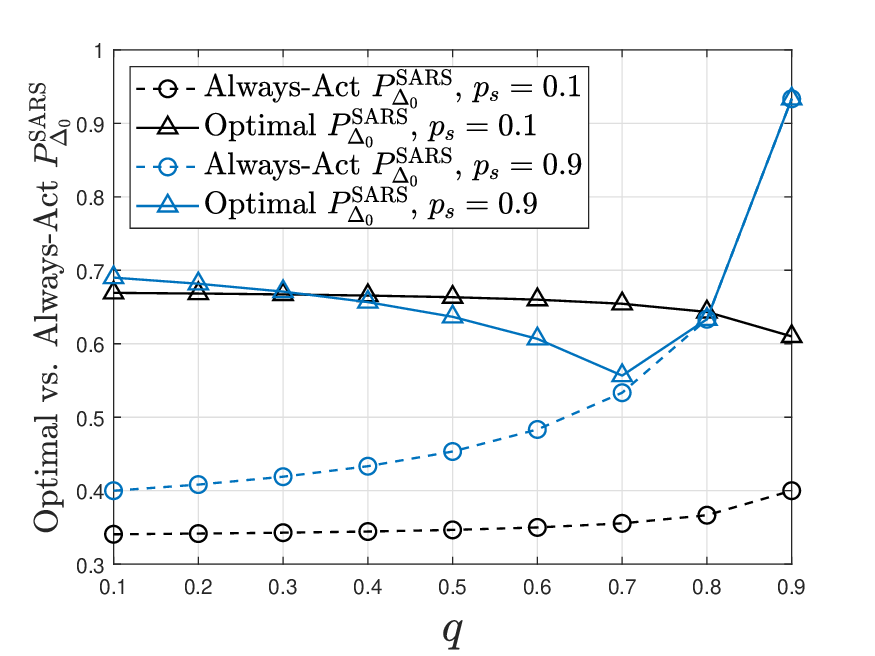}
  \label{PD0SARS_eta0.1}
		}
		\subfigure[$\eta = 0.7$]{\centering
			\includegraphics[trim=0.3cm 0.05cm 1.1cm 0.6cm,width=0.48\linewidth, clip]{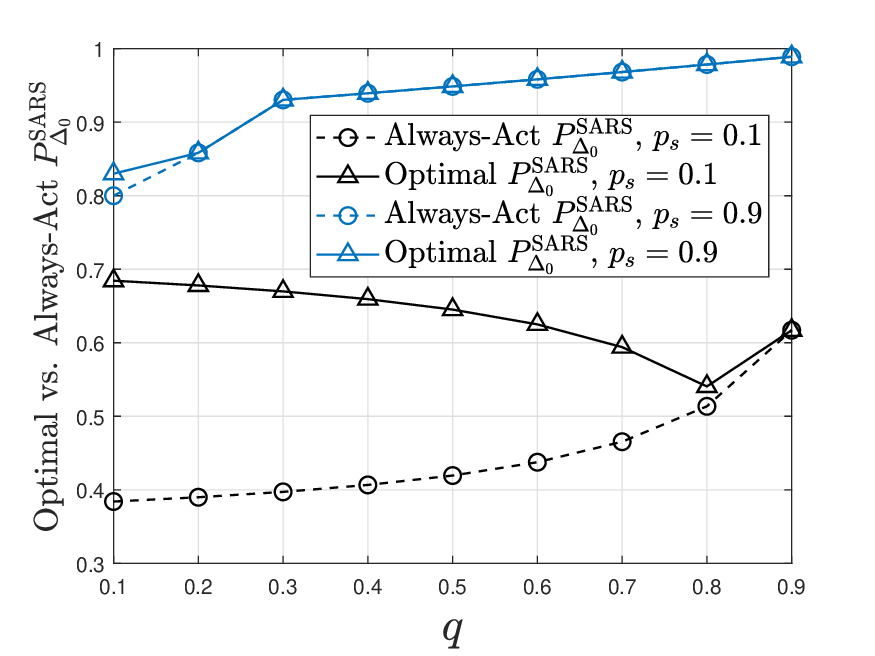}
		}
		\caption{Comparison between the optimal and Always-Act
$P_{\Delta_{0}}$ when the source is sampled using the optimal SARS policy as a function of $q$ for $\mu = 1$.}
		\label{PD0SARS}
	\end{figure}
     \begin{figure}[!t]
		\centering
 
		\subfigure[$\eta = 0.1$ ]{\includegraphics[trim=0.3cm 0.05cm 1.1cm 0.6cm,width=0.48\linewidth, clip]{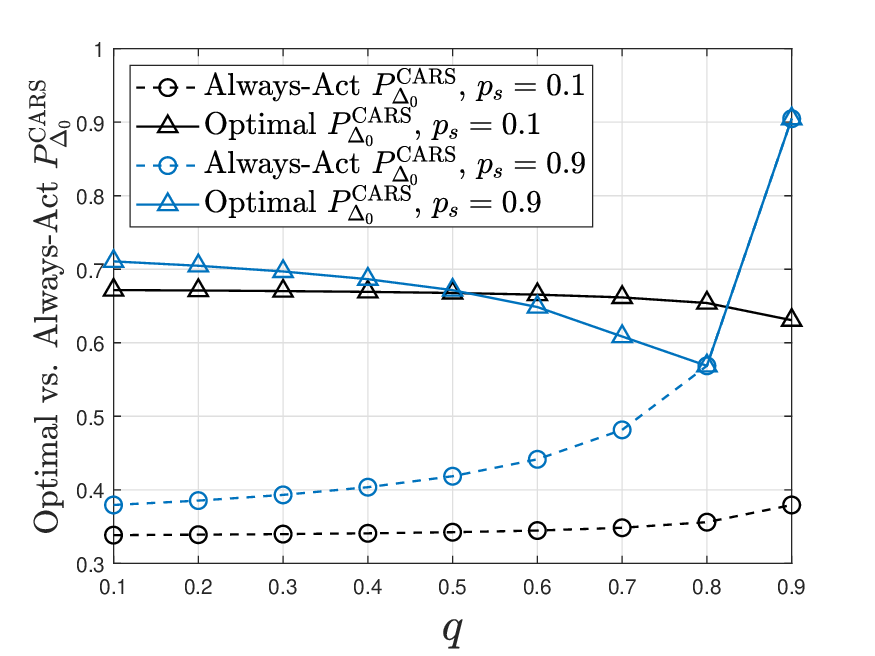}
  \label{PD0CARS_eta0.1}
		}
		\subfigure[$\eta = 0.7$]{\centering
			\includegraphics[trim=0.3cm 0.05cm 1.1cm 0.6cm,width=0.48\linewidth, clip]{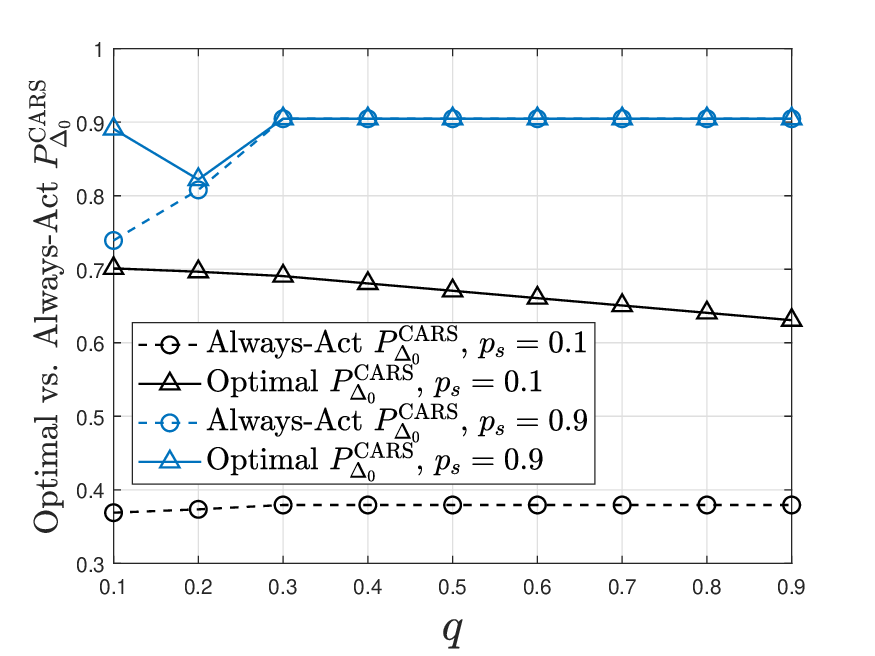}
		}
		\caption{Comparison between the optimal and Always-Act
$P_{\Delta_{0}}$ when the source is sampled using the optimal CARS policy as a function of $q$ for $\mu = 1$.}
		\label{PD0CARS}
	\end{figure}
    \section{Conclusion}
   We studied real-time monitoring of an $N$-state Markov process in a time-slotted system and derived closed-form expressions for time-averaged reconstruction error and the average AoII under several sampling and transmission policies. We also proposed a randomized actuation policy and derived the probability of correct actuation decisions. Three rate-constrained optimization problems were formulated to optimize sampling and actuation decisions jointly. Our results demonstrate that the semantics-aware randomized stationary policy outperforms the other policies for rapidly changing sources. In contrast, the threshold-aware randomized stationary policy is more effective for slowly changing sources under limited sampling and transmission budgets. Moreover, the proposed randomized actuation policy substantially reduces incorrect actuation decisions and enhances overall system performance. These results demonstrate that effective real-time monitoring requires the joint design of sampling, transmission, and actuation policies, as accurate state estimation alone is insufficient to guarantee reliable actuation decisions.
    \appendix
	\subsection{Proof of Lemma {\ref{theorem_InstError}}}
	\label{Appendix_LemmaPij_Et}
 For an $N$-state DTMC information source, the probability that $E(t) = i$ for $i \geqslant 0$ can be expressed as follows:
    \begin{align}
    \label{dist_RS}
        \mathbb{P}\big[E(t) = i\big] &= \mathbb{P}\big[X(t)-\hat{X}(t)=\pm i\big]\notag\\
        &=\begin{cases}
            \displaystyle{\sum_{k=0}^{N-1} \pi_{k,k}},& i=0,\\
            \displaystyle{\sum_{k=0}^{N-1-i} \pi_{k,i+k}+\sum_{m=0}^{N-1-i} \pi_{i+m,m}},&i\neq 0.
        \end{cases}
    \end{align}
    Note that $\pi_{i,j}$ in \eqref{dist_RS} are the probabilities obtained from the stationary distribution of the two-dimensional DTMC describing the joint state of the original and reconstructed source, i.e., $\big(X(t),\hat{X}(t)\big)$, where the transition probabilities $P_{m,n/m',n'}=\mathbb{P}[X(t+1)=m', \hat{X}(t+1) = n'|X(t)=m,\hat{X}(t)=n]$, $\forall m,n,m',n' \in\{0, 1, 2, \cdots, N-1\}$ for the RS policy are given by:
\begin{align}
    \label{transprob_k}
   & P_{m,n/m',n'}=\notag\\ 
    &\begin{cases}
        \!q,&\!\!\!\!\!\! m\!=\!n\!=\!m'\!=\!n',\\
        \!p, &\!\!\!\!\!\! m\!\neq\! n, m'\!=\!n'\!=\!n,\\
        \!q(1\!-\!p^{r}_{\alpha}p_{s}),&\!\!\!\!\!\! m\neq n, m'\!=\!m, n'\!=\!n,\\
        \!qp^{r}_{\alpha}p_{s},&\!\!\!\!\!\! m\!\neq\! n, m'\!=\!m\!=\!n',\\
        \!p(1\!-\!p^{r}_{\alpha}p_{s}),&\!\!\!\!\!\! \big\{m\!=\!n\!=\!n', m'\!\neq\! m\big\} \text{or} \big\{m'\!\neq\! m\!\neq\! n, n'\!=\!n\big\},\\
        \!pp^{r}_{\alpha}p_{s},&\!\!\!\!\!\!\big\{m\!=\!n, m'\!=\!n'\!\neq\! m\big\} \text{or} \big\{m'\!\neq\! m\!\neq\! n, m'\!\!=\!\!n'\!\big\},\\
        \!0, &\!\!\!\!\!\! \text{otherwise}.
    \end{cases}
\end{align}
Now, the stationary distribution can be obtained by solving \emph{balance equations} as follows:
\begin{align}
    \label{BalanceEq}
    \pi P_{I} = \pi, \hspace{0.2cm}\sum_{i=0}^{N-1}\sum_{j=0}^{N-1}\pi_{i,j} = 1,
\end{align}
where $P_{I}$ in \eqref{BalanceEq} represents the transition probability matrix, with elements defined in \eqref{transprob_k}. Additionally, $\pi$ denotes the stationary distribution, represented as a row vector $\pi = [\pi_{0,0},\!\cdots\!,\pi_{0,N\!-\!1}, \pi_{1,0},\!\cdots\!, \pi_{1,N\!-\!1},\!\cdots\!,\pi_{N-1,0},\!\cdots\!,\pi_{N-1,N-1}]$. Now, using \eqref{transprob_k} and \eqref{BalanceEq}, and noting that $q = 1 - (N-1)p$, the steady-state probabilities $\pi_{i,j}$, $\forall i,j \in \{0,1,\dots,N-1\}$, for the RS policy can be calculated as follows:
\begin{align}
\label{pij_RS}
            \pi_{i,j} =
            \begin{cases}
                \frac{p+(1-p)p^{r}_{\alpha}p_{s}}{Np^{r}_{\alpha}p_{s}+N^{2}p(1-p^{r}_{\alpha}p_{s})}, i=j,\\
                \frac{p(1-p^{r}_{\alpha}p_{s})}{Np^{r}_{\alpha}p_{s}+N^{2}p(1-p^{r}_{\alpha}p_{s})}, i\neq j.
            \end{cases}
        \end{align}
    Using \eqref{pij_RS}, \eqref{dist_RS} for the RS policy can be written as:
    \begin{align}
    \label{dist_RS2}
        \mathbb{P}\big[E(t) \!=\! i\big] \!\!=\!\!\begin{cases}
          \frac{p+(1-p)p^{r}_{\alpha}p_{s}}{p^{r}_{\alpha}p_{s}+Np(1-p^{r}_{\alpha}p_{s})}, &\!\!\!\!i \!=\!0,\\
            \frac{2p(N-i)(1-p^{r}_{\alpha}p_{s})}{Np^{r}_{\alpha}p_{s}+N^{2}p(1-p^{r}_{\alpha}p_{s})}, &\!\!\!\!i\!\neq\! 0, N\!>\!i,\\
            0, &\!\!\!\!i\neq 0,2\!\leqslant\! N \!\leqslant\! i.
        \end{cases}
    \end{align}
    Now, using \eqref{dist_RS2}, we can calculate the time-averaged reconstruction error under the RS policy as:
    \begin{align}
    \label{avgE_RS_proof}
        \bar{E}^{\text{RS}}\! =\! \sum_{i=1}^{N-1}i\mathbb{P}\big[E(t)\!=\!i\big] \!=\! \frac{p(N^{2}-1)(1-p^{r}_{\alpha^{s}}p_{s})}{3\big[p^{r}_{\alpha}p_{s}+Np(1-p^{r}_{\alpha}p_{s})\big]}.
    \end{align}
    Similarly, for the SARS policy, we can obtain \eqref{pij_RS} as follows:
    \begin{align}
        \label{pij_SARS}
         \pi_{i,j} =
            \begin{cases}
                \frac{p+(1-p)q_{\alpha_{2}}p_{s}}{N^{2}p-N(N-1)pq_{\alpha_{1}}p_{s}+N(1-p)q_{\alpha_{2}}p_{s}},& i=j,\\
                \frac{p(1-q_{\alpha_{1}p_{s}})}{N^{2}p-N(N-1)pq_{\alpha_{1}}p_{s}+N(1-p)q_{\alpha_{2}}p_{s}}, &i\neq j.
            \end{cases}
    \end{align}
     Using \eqref{pij_SARS}, the probability that $E(t)$ equals $i\geqslant 0$ for the SARS policy is given by:
    \begin{align}
    \label{PE_SA}
        &\mathbb{P}\big[E(t) \!=\! i\big] \notag\\
        &=\!\begin{cases}
        \frac{p+(1-p)q_{\alpha_{2}}p_{s}}{Np-(N-1)pq_{\alpha_{1}}p_{s}+(1-p)q_{\alpha_{2}}p_{s}},&\!\!\!i\!=\!0,\\
      \frac{2(N-i)p(1-q_{\alpha_{1}}p_{s})}{N^{2}p-N(N-1)pq_{\alpha_{1}}p_{s}+N(1-p)q_{\alpha_{2}}p_{s}}, &\!\!\!i\!\neq\! 0, N\!>\!i,\\
            0, &\!\!\!i\!\neq\! 0, 2\!\leqslant\! N \!\leqslant\! i.
        \end{cases}
    \end{align}
    Using \eqref{PE_SA}, the time-averaged reconstruction error under the SARS policy can be calculated as:
    \begin{align}
    \label{avgE_SARS}
        \bar{E}^{\text{SARS}} &= \sum_{i=0}^{N-1}i\mathbb{P}\big[E(t)=i\big] \notag\\
        &=\! \frac{p(1-q_{\alpha_{1}}p_{s})(N^{2}-1)}{3\big[Np-(N-1)pq_{\alpha_{1}}p_{s}+(1-p)q_{\alpha_{2}}p_{s}\big]}.
    \end{align}
   Furthermore, for the CARS policy, \eqref{pij_RS} can be written as follows:
    \begin{align}
        \label{pij_CARS}
         \pi_{i,j} =
            \begin{cases}
                \frac{1+(N-2)p^{c}_{\alpha}p_{s}}{N(N-p^{c}_{\alpha}p_{s})}, &i=j,\\
                \frac{1-p^{c}_{\alpha}p_{s}}{N(N-p^{c}_{\alpha}p_{s})}, &i\neq j.
            \end{cases}
    \end{align}
   Using \eqref{pij_CARS}, the probability that $E(t) = i$ for $i \geqslant 0$ under the CARS policy is given by:
    \begin{align}
    \label{PE_CARS}
        \mathbb{P}\big[E(t) = i\big] &=\begin{cases}
        \frac{1+(N-2)p^{c}_{\alpha}p_{s}}{N-p^{c}_{\alpha}p_{s}},& i =0,\\
            \frac{2(N-i)(1-p^{c}_{\alpha}p_{s})}{N(N-p^{c}_{\alpha}p_{s})}, &i\neq 0, N>i,\\
            0, &i\neq 0,2\leqslant N \leqslant i.
        \end{cases}
    \end{align}
 Using \eqref{PE_CARS}, the time-averaged reconstruction error under the CARS policy can be derived as:
    \begin{align}
    \label{avgE_CARS}
        \bar{E}^{\text{CARS}} = \sum_{i=0}^{N-1}i\mathbb{P}\big[E(t)=i\big] = \frac{(N^{2}-1)(1-p^{c}_{\alpha}p_{{s}})}{3(N-p^{c}_{\alpha}p_{{s}})}.
    \end{align}
    \subsection{Proof of Lemma {\ref{theorem_ConsecutiveError}}}
    \label{Appendix_Lemma_CE}
    \begin{figure}[!t]
	\centering
	\begin{tikzpicture}[start chain=going left,->,>=latex,node distance=1.65cm,on grid,auto]
 \footnotesize
		\node[on chain]                        (h2) {$\cdots$};
	   \node[state, on chain]              (h1) {$n$};
       \node[state, on chain]              (h0) {\scriptsize{$n\!\!-\!\!1$}};
		\node[on chain]               (g) {$\cdots$};
		\node[state, on chain]                 (2) {$1$};
		\node[state, on chain]                 (1) {$0$};
		\draw[>=latex]
		(1)   edge[loop above] node {\scriptsize{$q$}}   (1)
		(1) edge  [bend left=20] node {\scriptsize{$(N\!-\!1)p$}} (2)
		(2) edge  [bend left=20] node[above] {\scriptsize{$p$}} (1)
        (2) edge  [bend left=20] node[above] {\scriptsize{$1-p$}} (g)
		(g) edge  [bend left=33] node[above,pos=0.3] {\scriptsize{$p$}} (1)
        (h0) edge  [bend left=45] node[above] {\scriptsize{$P_{n-1,0}$}} (1)
        (h0) edge  [bend left=20] node[above] {\scriptsize{$P_{n-1,n}$}} (h1)
		(h1) edge  [bend left=56] node[above] {\scriptsize{$P_{n,0}$}} (1)
		(h1) edge  [bend left=20] node[above] {\scriptsize{$P_{n,n+1}$}} (h2)
		;
	\end{tikzpicture}
	\vspace*{-2ex}
	\caption{DTMC model for the evolution of the AoII under the TARS policy, where $q=1-(N-1)p$, $P_{n-1,0}=p+(1\!-\!p)p^{th}_{\alpha}p_{s}$, $P_{n-1,n}=(1\!-\!p)(1\!-\!p^{th}_{\alpha}p_{s})$, $P_{n,0}=p+(1-p)p_{s}$, and $P_{n,n+1}=(1-p)(1-p_{s})$.}
	\label{Fig_AoII}
\end{figure}
   Using the DTMC model that describes the evolution of the AoII metric under the TARS policy, as shown in Fig.~\ref{Fig_AoII}, the steady-state probabilities $\pi_{i}$, $ \forall i \in \{1,2,\dots\}$, for the TARS policy can be calculated as:
        \begin{align}
        \label{pi_CE_n_proof}
            &\pi_{i} =\notag\\
           & \begin{cases}
               (N\!-\!1)p(1\!-\!p)^{i-1}\pi_{0}, &\!\!\!1\leqslant i\leqslant n\!-\!1,\\
               (N\!-\!1)p(1\!-\!p^{th}_{\alpha}p_{s})(1\!-\!p)^{i-1}(1\!-\!p_{s})^{i-n}\pi_{0},&\!\!\!i\geqslant n,
            \end{cases}
        \end{align}
where $\pi_{0}$ is given by:
\begin{align}
    \label{pi0_CE_n_Proof}
    \pi_{0} \!=\! \frac{p+(1-p)p_{s}}{N\big(p\!+\!(1\!-\!p)p_{s}\big)\!-\!(N\!-\!1)\big(1\!-\!p(1\!-\!p^{th}_{\alpha})\big)p_{s}(1-p)^{n-1}}.
\end{align}
Using \eqref{pi_CE_n_proof}, the average AoII under the TARS policy can be obtained as follows:
\begin{align}
        \label{AvgCE_n_proof}
        \overline{\text{AoII}} &= \sum_{i=1}^{\infty}i\pi_{i}=\frac{F(p^{th}_{\alpha},n)}{G(p^{th}_{\alpha},n)},
    \end{align}
   where $F(p^{th}_{\alpha},n)$ and $G(p^{th}_{\alpha},n)$ are given by:
    \begin{align}
    \label{FG_threhold_proof}
        F(p^{th}_{\alpha},n)&\!=\!(N-1)\Big[(1-p)\big(p\!+\!(1-p)p_{s}\big)^{2}\!+\!(1-p)^{n}p_{s}\notag\\
        &\!\times\!\Big(3p^{2}\!\!-\!\!2p\!-\!np^{2}\!-\!p^{3}\!+\!np^{3}(1\!-\!p^{th}_{\alpha})\!-\!p^{2}p^{th}_{\alpha}\!+\!p^{3}p^{th}_{\alpha}\notag\\
        &\!+\!(1-p)p_{s}\big(1\!+\!(n-2)p\!-\!(n-1)p^{2}(1-p^{th}_{\alpha})\big)\Big)\Big],\notag\\
        G(p^{th}_{\alpha},n)&\!=\!p\big(p_{s}+(1-p_{s})p\big)\big[N(1-p)\big(p_{s}+(1-p_{s})p\big)\notag\\
        &-(N-1)(1-p)^{n}\big(1-(1-p^{th}_{\alpha})p\big)p_{s}\big].
    \end{align}
Similarly, for the RS policy, \eqref{pi_CE_n_proof} can be obtained as follows:
 \begin{align}
\label{pi_CE_RS_proof}
    \!\!\pi_{i}\!\!=\!\!
    \begin{cases}
        \!\!\frac{p+(1-p)p^{r}_{\alpha}p_{s}}{Np +(1-Np) p^{r}_{\alpha}p_{s}}, & \!\!\!i=0,\\
        \!\big[(1-p)(1-p^{r}_{\alpha}p_{s})\big]^{i-1}\!\big[(N-1)p(1-p^{r}_{\alpha}p_{s})\big]\pi_{0},&\!\!\! i\geqslant 1.
    \end{cases}
\end{align}
Using \eqref{pi_CE_RS_proof}, the average AoII for the RS policy is given by:
     \begin{align}
        \label{AvgCE_RS_proof}
        \overline{\text{AoII}} \!=\!\frac{(N-1)p(1-p^{r}_{\alpha}p_{s})}{\big[p+(1-p)p^{r}_{\alpha}p_{s}\big]\big[Np+(1-Np)p^{r}_{\alpha}p_{s}\big]}.
    \end{align}
Moreover, under the SARS policy, the steady-state probabilities $\pi_{i}$ for all $i \geqslant 0$ are given by:
\begin{align}
\label{pi_CE_SARS_proof}
    \!\!\pi_{i}\!=\!
    \begin{cases}
        \frac{p+(1-p)q_{\alpha_{2}}p_{s}}{Np +p(1-N)q_{\alpha_{1}}p_{s}+(1-p)q_{\alpha_{2}}p_{s}}, & \!\!i=0,\\
        \big[(1\!-\!p)(1\!-\!q_{\alpha_{2}}p_{s})\big]^{i-1}\big[(N\!-\!1)p(1\!-\!q_{\alpha_{1}}p_{s})\big]\pi_{0},& \!\!i\geqslant 1.
    \end{cases}
\end{align}
Using \eqref{pi_CE_SARS_proof}, the average AoII under the SARS policy can be calculated as follows:
 \begin{align}
        \label{AvgCE_CARS_proof}
        \overline{\text{AoII}} \!\!=\!\!\frac{(N\!-\!1)p(1\!-\!q_{\alpha_{1}}p_{s})}{\big[p\!+\!(1\!-\!p)q_{\alpha_{2}}p_{s}\big]\big[Np\!+\!(1\!\!-\!\!N)pq_{\alpha_{1}}p_{s}\!+\!(1\!-\!p)q_{\alpha_{2}}p_{s}\big]}.
    \end{align}
Furthermore, for the CARS policy, the steady-state probabilities $\pi_i$ for all $i \geqslant 0$ can be obtained as follows:
    \begin{align}
\label{pi_CE_CARS_proof}
    &\pi_{i}\notag\\
    &\!=\!
    \begin{cases}
        \!\frac{1+(N-2)p^{c}_{\alpha}p_{s}}{N-p^{c}_{\alpha}p_{s}}, &\!\!\! i=0,\\
        \!\big[1\!-\!p\big(1\!+\!(N\!-\!2)p^{c}_{\alpha}p_{s}\big)\big]^{i\!-\!1}\!\big[(N\!\!-\!\!1)p(1\!\!-\!\!p^{c}_{\alpha}p_{s})\big]\pi_{0},&\!\!\! i\geqslant 1.
    \end{cases}
\end{align}
Using \eqref{pi_CE_CARS_proof}, the average AoII under the CARS policy is given by:
 \begin{align}
        \label{AvgCE_CARS_proof}
        \overline{\text{AoII}} = \sum_{i=1}^{\infty}i\pi_{i}=\frac{(N-1)(1-p^{c}_{\alpha}p_{s})}{p(N-p^{c}_{\alpha}p_{s})\big[1+(N-2)p^{c}_{\alpha}p_{s}\big]}.
    \end{align}
    \subsection{Proof of Lemma {\ref{theorem_CoIA}}}	\label{Appendix_ControlIncorrectAction}
Using \eqref{CoIA} and the total probability theorem, the probability 
  $\mathbb{P}\big[\Delta(t)=0]$ can be written as follows:
        \begin{align}
           \label{PD0}  &\!\!\scalebox{0.9}{$\mathbb{P}\big[\Delta(t)=0\big]$}\notag\\
&\!\!=\scalebox{0.9}{$\mathbb{P}\big[c(t)=1\big|X(t)\!=\!\hat{X}(t),\alpha(t)\!=\!0\big]\mathbb{P}\big[X(t)\!=\!\hat{X}(t),\alpha(t)\!=\!0\big]$}\notag\\
           &\!\!+\scalebox{0.9}{$\mathbb{P}\big[c(t)=1\big|X(t)=\hat{X}(t),\alpha(t)=1,h(t)=0\big]$}\notag\\
           &\!\!\times\scalebox{0.9}{$\mathbb{P}\big[X(t)=\hat{X}(t),\alpha(t)=1,h(t)=0\big]$}\notag\\
           &\!\!+\scalebox{0.9}{$\mathbb{P}\big[c(t)=1\big|X(t)=\hat{X}(t),\alpha(t)=1,h(t)=1\big]$}\notag\\
           &\!\!\times\scalebox{0.9}{$\mathbb{P}\big[X(t)=\hat{X}(t),\alpha(t)=1,h(t)=1\big]$}\notag\\
           &\!\!+\scalebox{0.9}{$\mathbb{P}\big[c(t)=0\big|X(t)\!\neq\!\hat{X}(t),\alpha(t)\!=\!0\big]\mathbb{P}\big[X(t)\!\neq\!\hat{X}(t),\alpha(t)\!=\!0\big]$}\notag\\&\!\!+\scalebox{0.9}{$\mathbb{P}\big[c(t)=0\big|X(t)\neq\hat{X}(t),\alpha(t)=1,h(t)=0\big]$}\notag\\
           &\!\!\times\scalebox{0.9}{$\mathbb{P}\big[X(t)\neq\hat{X}(t),\alpha(t)=1,h(t)=0\big]$}\notag\\
           &\!\!=\scalebox{0.9}{$p_{c_{2}}\Big(\!\mathbb{P}\big[X(t)\!\!=\!\!\hat{X}(t),\alpha(t)\!=\!0\big]\!\!+\!\!\mathbb{P}\big[X(t)\!\!=\!\!\hat{X}(t),\alpha(t)\!=\!1,h(t)\!=\!0\big]\!\Big)$}\!\!\notag\\
           &\!\!+\!\!\scalebox{0.9}{$(1\!-\!p_{c_{2}})\!\Big(\!\mathbb{P}\big[X(t)\!\!\neq\!\!\hat{X}(t),\alpha(t)\!\!=\!\!0\big]\!\!+\!\!\mathbb{P}\big[X(t)\!\!\neq\!\!\hat{X}(t),\alpha(t)\!\!=\!\!1,h(t)\!\!=\!\!0\big]\!\Big)$}\notag\\
           &\!\!+\scalebox{0.9}{$p_{c_{1}}\mathbb{P}\big[X(t)\!\!=\!\!\hat{X}(t),\alpha(t)\!=\!1,h(t)\!=\!1\big]$},
        \end{align}
      When the source is sampled using the RS policy, the probabilities in \eqref{PD0} can be obtained as follows:
        \begin{align}
            &\mathbb{P}\big[X(t)=\hat{X}(t),\alpha(t)=0\big]\notag\\
            &= \sum_{i=0}^{N-1}\Big(\mathbb{P}\big[X(t)\!=\!\hat{X}(t),\alpha(t)\!=\!0\big|X(t-1)\!=\!i,\hat{X}(t-1)\!=\!i\big]\notag\\
            &\hspace{2cm}\times \mathbb{P}\big[X(t-1)\!=\!i,\hat{X}(t-1)\!=\!i\big]\Big)\notag\\
            &+\!\!\sum_{i=0}^{N-1}\!\sum_{\substack{j=0 \\ j \neq i}}^{N-1}\!\Big(\mathbb{P}\big[X(t)\!\!=\!\!\hat{X}(t),\alpha(t)\!\!=\!\!0\big|X(t\!-\!1)\!=\!i,\hat{X}(t\!-\!1)\!=\!j\big]\notag\\
            &\times \!\!\mathbb{P}\big[\!X(t\!\!-\!\!1)\!\!=\!i,\hat{X}(t\!-\!1)\!\!=\!\!j\big]\!\Big)\!\!=\!\!\frac{(1\!-\!p^{r}_{\alpha})\big(p+(1-Np)p^{r}_{\alpha}p_{s}\big)}{p^{r}_{\alpha}p_{s}\!+\!Np(1\!-\!p^{r}_{\alpha}p_{s})}.\label{Cond1_RS}
        \end{align}
       Using a procedure similar to that in \eqref{Cond1_RS}, one can obtain the probabilities $\mathbb{P}\big[X(t)=\hat{X}(t),\alpha(t)=1,h(t)=0\big]$, $\mathbb{P}\big[X(t)=\hat{X}(t),\alpha(t)=1,h(t)=1\big]$, $\mathbb{P}\big[X(t)\neq\hat{X}(t),\alpha(t)=0\big]$, $\mathbb{P}\big[X(t)\neq\hat{X}(t),\alpha(t)=1,h(t)=0\big]$ as follows:
        \begin{align}
  &\mathbb{P}\big[X(t)=\hat{X}(t),\alpha(t)=1,h(t)=0\big]\notag\\
  &= \frac{p^{r}_{\alpha}(1-p_{s})(p+(1-Np)p^{r}_{\alpha}p_{s})}{p^{r}_{\alpha}p_{s}\!+\!Np(1\!-\!p^{r}_{\alpha}p_{s})},\label{Cond2_RS}\\
            &\mathbb{P}\big[X(t)=\hat{X}(t),\alpha(t)=1,h(t)=1\big]=p^{r}_{\alpha}p_{s},\label{Cond3_RS}\\
            &\mathbb{P}\big[X(t)\neq\hat{X}(t),\alpha(t)=0\big]=\frac{(N-1)p(1-p^{r}_{\alpha})}{Np+(1-Np)p^{r}_{\alpha}p_{s}},\label{Cond4_RS}\\
            &\mathbb{P}\big[X(t)\!\neq\!\hat{X}(t),\alpha(t)\!=\!1,h(t)\!=\!0\big]\!=\!\frac{(N\!-\!1)pp^{r}_{\alpha}(1\!-\!p_{s})}{Np\!+\!(1\!-\!Np)p^{r}_{\alpha}p_{s}}.\label{Cond5_RS}
        \end{align}
  Now, using \eqref{PD0} and eqs.~\eqref{Cond1_RS}--\eqref{Cond5_RS}, the probability $\mathbb{P}\big[\Delta(t)=0\big]$ when the source is sampled using the RS policy is given by:
         \begin{align}
           \label{Delta_RS_proof}
           P^{\text{RS}}_{\Delta_{0}}&=\frac{(N-1)p(1-p^{r}_{\alpha}p_{s})}{Np+(1-Np)p^{r}_{\alpha}p_{s}}+p^{r}_{\alpha}p_{s}p_{c_{1}}\notag\\
           &+\frac{(1-p^{r}_{\alpha}p_{s})\big[2p-Np+(1-Np)p^{r}_{\alpha}p_{s}\big]p_{c_{2}}}{Np+(1-Np)p^{r}_{\alpha}p_{s}}.
        \end{align}
         \subsection{Proof of Lemmas \ref{theorem_Optimization_AvgE_SARS}, \ref{theorem_Optimization_AvgE_CARS}, and \ref{theorem_Optimization_AvgCost_THR}}
         \label{AvgSamplinCost_Proof}
     For the SARS policy, the constraint in \eqref{OptProb_AvgE_SARS_Const} can be written as follows:      
        \begin{align}
\label{AvgSamplingCost_SARS_proof}
 &\lim_{T \to \infty}\frac{1}{T}\mathbb{E}\Bigg[\sum_{t=1}^{T} \mathbbm{1}\{\alpha(t)=1\}\Bigg]\notag\\
 &=\mathbb{P}\big[\alpha(t)=1\big|X(t)\neq\hat{X}(t-1),X(t-1)=\hat{X}(t-1)\big]\notag\\
 &\times\mathbb{P}\big[X(t)\neq\hat{X}(t-1),X(t-1)=\hat{X}(t-1)\big]\notag\\
 &+\mathbb{P}\big[\alpha(t)=1\big|X(t)\neq\hat{X}(t-1),X(t-1)\neq\hat{X}(t-1)\big]\notag\\
 &\times\mathbb{P}\big[X(t)\neq\hat{X}(t-1),X(t-1)\neq\hat{X}(t-1)\big]\notag\\
 &= q_{\alpha_{1}}N(N-1)p\pi_{i,i}+ q_{\alpha_{2}}N(N-1)(1-p))\pi_{i,j},
        \end{align} 
        Using \eqref{pij_SARS}, \eqref{AvgSamplingCost_SARS_proof} is simplified as:
         \begin{align}
\label{AvgSamplingCost_SARS2}
 &\lim_{T \to \infty}\frac{1}{T}\mathbb{E}\Bigg[\sum_{t=1}^{T} \mathbbm{1}\{\alpha(t)=1\}\Bigg]\notag\\
 &=\frac{ p(N-1)\big[p(q_{\alpha_{1}}-q_{\alpha_{2}})+q_{\alpha_{2}}\big]}{Np-(N-1)pq_{\alpha_{1}}p_{s}+(1-p)q_{\alpha_{2}}p_{s}}.
        \end{align} 
       Furthermore, the time-averaged sampling frequency for the CARS policy can be calculated as follows:
         \begin{align}
\label{AvgSamplingCost_CARS}
 &\lim_{T \to \infty}\frac{1}{T}\mathbb{E}\Bigg[\sum_{t=1}^{T} \mathbbm{1}\{\alpha(t)=1\}\Bigg]\notag\\
     &= \mathbb{P}\big[\alpha(t)=1\big|X(t)\neq X(t-1)\big]\mathbb{P}\big[X(t)\neq X(t-1)\big]\notag\\
 &= p^{c}_{\alpha}\sum^{N-1}_{i=0}\mathbb{P}\big[X(t)\!\neq\! X(t-1)\big|X(t-1)\!=\!i\big]\mathbb{P}\big[X(t-1)\!=\!i\big]\notag\\
 &=(N-1)pp^{c}_{\alpha}.
        \end{align}
       Moreover, the time-averaged sampling frequency under the TARS policy is given by:
\begin{align}
    \label{AvgCost_THR}
    \lim_{T \to \infty}\frac{1}{T}\mathbb{E}\Bigg[\sum_{t=1}^{T} \mathbbm{1}\{\alpha(t)=1\}\Bigg] \!=\!  p^{th}_{\alpha}\pi_{n-1}\!+\!\sum_{i=n}^{\infty}\pi_{i},
\end{align}
Using \eqref{pi_CE_n_proof} and \eqref{pi0_CE_n_Proof}, we can write \eqref{AvgCost_THR} as follows:
\begin{align}
         &\lim_{T \to \infty}\frac{1}{T}\mathbb{E}\Bigg[\sum_{t=1}^{T} \mathbbm{1}\{\alpha(t)=1\}\Bigg]\notag\\
         &\!=\!
         \begin{cases}
             \frac{p(N-1)+\big(p-(Np-1)p_{s}\big)p^{th}_{\alpha}}{Np+(1-p)p_{s}-(N-1)pp^{th}_{\alpha}p_{s}},\hspace{1.7cm} &n=1,\notag\\
             \frac{(N-1)p(1-p)^{n-1}\big(1-(1-p^{th}_{\alpha})p\big)}{N(1\!-\!p)\big(p+(1\!-\!p)p_{s}\big)-(N\!-\!1)(1\!-\!p)^{n}\big(1-(1-p^{th}_{\alpha})p\big)p_{s}},&n\geqslant 2.
         \end{cases}
    \end{align}
        \subsection{Proof of Lemmas {\ref{theorem_Optimization_PD0_RS}, \ref{theorem_Optimization_PD0_SARS}, and \ref{theorem_Optimization_PD0_CARS}}}	\label{Appendix_CoIA_TAvgSamplingCost}
       Using the total probability theorem, the constraint in \eqref{Optimization_prob_PD0_constraint1} can be written as:
        \begin{align}
		\label{AvgSCost_CoIA_RS_proof}
		&\lim_{T \to \infty}\frac{1}{T}\mathbb{E}\Bigg[\sum_{t=1}^{T} \mathbbm{1}\{c(t)=1\}\Bigg]\notag\\
    &=\mathbb{P}\big[c(t)\!=\!1\big|X(t)\!=\!\hat{X}(t),\alpha(t)\!=\!0\big]\mathbb{P}\big[X(t)\!=\!\hat{X}(t),\alpha(t)\!=\!0\big]\notag\\
    &+\mathbb{P}\big[c(t)=1\big|X(t)=\hat{X}(t),\alpha(t)=1,h(t)=0\big]\notag\\
    &\times\mathbb{P}\big[X(t)=\hat{X}(t),\alpha(t)=1,h(t)=0\big]\notag\\
    &+\mathbb{P}\big[c(t)=1\big|X(t)=\hat{X}(t),\alpha(t)=1,h(t)=1\big]\notag\\
    &\times\mathbb{P}\big[X(t)=\hat{X}(t),\alpha(t)=1,h(t)=1\big]\notag\\
    &+\mathbb{P}\big[c(t)=1\big|X(t)\neq \hat{X}(t)\big]\mathbb{P}\big[X(t)=\hat{X}(t)\big],
	\end{align}
  Using eqs.~\eqref{Cond1_RS}--\eqref{Cond5_RS}, the time-averaged actuation frequency when the source is sampled using the RS policy is given by:
    \begin{align}
		\label{AvgSCost_CoIA_RS2_proof}
		&\lim_{T \to \infty}\frac{1}{T}\mathbb{E}\Bigg[\sum_{t=1}^{T}\mathbbm{1}\{c(t)\!=\!1\}\Bigg]\!=\!p^{r}_{\alpha}p_{s}p_{c_{1}}\!\!+\!\!(1\!-\!p^{r}_{\alpha}p_{s})p_{c_{2}}.
	\end{align}
Furthermore, the time-averaged actuation frequency when the source is sampled using the SARS policy is calculated as follows:
    \begin{align}
		\label{AvgSCost_CoIA_SARS_proof}
		&\lim_{T \to \infty}\frac{1}{T}\mathbb{E}\Bigg[\sum_{t=1}^{T}\mathbbm{1}\{c(t)=1\}\Bigg]\notag\\
        &=\Bigg(\frac{(N-1)pp_{s}\big[p(q_{\alpha_{1}}-q_{\alpha_{2}})+q_{\alpha_{2}}\big]}{Np+(1-N)pq_{\alpha_{1}}p_{s}+(1-p)q_{\alpha_{2}}p_{s}}\Bigg)p_{c_{1}}\notag\\
        &+\Bigg(\frac{Np + (1 - N)p (1 + p)  q_{\alpha_{1}}p_{s}}{Np+(1-N)pq_{\alpha_{1}}p_{s}+(1-p)q_{\alpha_{2}}p_{s}}\notag\\
        &+\frac{(1-p)\big(1+(1-N)p\big)q_{\alpha_{2}}p_{s}}{Np+(1-N)pq_{\alpha_{1}}p_{s}+(1-p)q_{\alpha_{2}}p_{s}}\Bigg)p_{c_{2}}.
	\end{align}
Moreover, when the source is sampled under the CARS policy, the time-averaged actuation frequency is given by:
\begin{align}
		\label{AvgSCost_CoIA_CARS_proof}
		&\lim_{T \to \infty}\frac{1}{T}\mathbb{E}\Bigg[\sum_{t=1}^{T} \mathbbm{1}\{c(t)=1\}\Bigg]\notag\\
        &=(N-1)pp^{c}_{\alpha}p_{s}p_{c_{1}}+\big(1+(1-N)pp^{c}_{\alpha}p_{s}\big)p_{c_{2}}.
	\end{align}
\bibliographystyle{IEEEtran}
\bibliography{ref}
\end{document}